%% file: bare_jrnl_RR1.tex
\documentclass[journal]{IEEEtran}

\ifCLASSINFOpdf
  \usepackage[pdftex]{graphicx}
\else  
   \usepackage[dvips]{graphicx}
\fi

\usepackage{amsmath,amsfonts}
\usepackage{amsthm}
\usepackage{algorithm}
\usepackage{array}
\usepackage[caption=true,font=normalsize,labelfont=sf,textfont=sf]{subfig}
\usepackage{textcomp}
\usepackage{stfloats}
\usepackage{url}
\usepackage{verbatim}
\usepackage{graphicx}
\usepackage{cite}
\usepackage{cuted}
\usepackage{setspace}
\usepackage{amssymb}
\usepackage{mathtools}
\usepackage{amsthm}
\usepackage{algpseudocode}
\usepackage{makecell}
\usepackage{caption}
\usepackage{float}
\usepackage{adjustbox}
\usepackage[table,x11names]{xcolor}
\usepackage{colortbl}
\usepackage{fancyhdr}
\usepackage{optidef}
\usepackage{siunitx}
\usepackage{stfloats}
\usepackage{color}
\usepackage{tikz}
\usepackage{pgfplots}
\usepackage{xcolor}

\newtheorem{theorem}{Lemma}
\newtheorem{prop}{Proposition}
\newtheorem{remark}{Remark}

\newcommand{\PK}{\color{blue}}

\newcommand\at[2]{\left.#1\right|_{#2}}
\DeclareMathOperator*{\argmax}{arg\,max}
\DeclareMathOperator*{\argmin}{arg\,min}

\pgfplotsset{compat=1.18}
\begin{document}

\title{Optimal Distortion-Aware Multi-User  Power Allocation for Massive MIMO Networks}

\author{Siddarth~Marwaha,~\IEEEmembership{Student Member,~IEEE,}
Pawel~Kryszkiewicz,~\IEEEmembership{Senior Member,~IEEE,}
        and~Eduard~Jorswieck,~\IEEEmembership{Fellow,~IEEE}
        \thanks{P. Kryszkiewicz is with the Institute of Radiocommunications, Poznan University of Technology, POLAND e-mail: pawel.kryszkiewicz@put.poznan.pl}%
        \thanks{S. Marwaha and E. Jorswieck are with the Institute of Communication Technology, University of Braunschweig, GERMANY email:\{s.marwaha, e.jorswieck\}@tu-braunschweig.de}%
        \thanks{The work of S. Marwaha was funded in part by BMDV 5G COMPASS project within InnoNT program under Grant 19OI22017A. Research of P. Kryszkiewicz was funded by the Polish National Science Center, project no. 2021/41/B/ST7/00136. For the purpose of Open Access, the author has applied a CC-BY public copyright license to any Author Accepted Manuscript (AAM) version arising from this submission. The work of E. Jorswieck was partially funded by BMBF 6G RIC project under Grant 16KISK031.
        }
        }

\maketitle

\begin{abstract}

Real-world wireless transmitter front-ends exhibit certain nonlinear behavior, e.g., signal clipping by a Power Amplifier (PA). Although many resource allocation solutions do not consider this for simplicity, it leads to inaccurate results or a reduced number of degrees of freedom, not achieving the global performance.
In this work, we propose an optimal PA distortion-aware power allocation strategy in a downlink orthogonal frequency division multiplex (OFDM) based massive multiple-input multiple-output (M-MIMO) system. 
Assuming a soft-limiter PA model, where the transmission occurs under small-scale independent and identically distributed (i.i.d) Rayleigh fading channel, we derive the wideband signal-to-noise-and-distortion ratio (SNDR) and formulate the power allocation problem. Most interestingly, the distortion introduced by the PA leads to an SNDR-efficient operating point without explicit transmit power constraints. 
While the optimization problem is non-convex, we decouple it into a non-convex total power allocation problem and a convex power distribution problem among the users (UEs). We propose an alternating optimization algorithm to find the optimum solution. Our simulation results show significant sum-rate gains over existing distortion-neglecting solutions, e.g., a median 4 times increase and a median 50\% increase for a 64-antenna and 512-antenna base station serving 60 users, respectively.

\end{abstract}


\IEEEpeerreviewmaketitle

\section{Introduction}

Massive multiple input multiple output (M-MIMO) has been shown to increase throughput and enable massive multiple access \cite{Bjornson_2018_MIMO_unlimited}, however, under the assumption that hundreds of parallel high-quality front-ends are employed. Not only does this require expensive hardware, but it also results in high energy consumption, which can be justified by considering a single Power Amplifier (PA) in the M-MIMO transmitter. Typically, a multi-carrier waveform is used for M-MIMO processing, however, this results in high variations of the instantaneous signal power \cite{ochiai2013analysis}, measured typically by the Peak-to-Average-Power Ratio (PAPR). Unfortunately, as such a signal passes through a practical PA, i.e., a signal with limited maximal output power, some of its samples can be distorted, introducing signal distortion in each transmission chain. From the Signal-to-Distortion Ratio (SDR) perspective, even for the optimal PA, i.e., a soft-limiter that provides linear amplification till its clipping level, all the samples exceeding the clipping threshold will generate a non-linear distortion signal.
Fortunately, the clipping problem can be reduced by means of certain transmitter-side signal processing, e.g., reduction of PA input signal variation \cite{Kryszkiewicz_TR_2018}, or advanced reception algorithms \cite{Dinis_Qoptimal_RX_2024}. Additionally, this effect can be reduced by optimizing the transmit power in relation to the PA clipping level. A higher signal-to-noise ratio (SNR) can be achieved by allocating higher transmit power, however, this results in significant signal clipping, thereby reducing the SDR. Therefore, considering the PA model, there is significant potential for wireless link optimization by distortion aware power allocation.  

A few solutions have been proposed to address this issue in Single Input Single Output (SISO) multi-carrier systems. In \cite{Baghani_IMD3_power_allocation_wideband_2014}, the authors focused on optimizing total transmit power in a cognitive radio system based on orthogonal frequency division multiplexing (OFDM), where the PA is modeled using a third-order memoryless polynomial. Moreover, equal power is allocated on each utilized sub-carrier. The authors assumed that the non-linear distortion power is equally distributed among sub-carriers, which is only approximately true as demonstrated in \cite{lee2014characterization}. A similar approach has been applied to a cognitive raido system based on generalized frequency division multiplexing (GFDM) in \cite{Amirhossein_2019_power_allocation_GFDM}. \textcolor{blue}{Furthermore, in \cite{10697405}, a similar third-order polynomial-based PA model—fitted to real-world measurements—was adopted to optimize transmission power for maximizing throughput in the uplink of a Non Orthogonal Multiple Access (NOMA) system. However, the third-order polynomial PA model is relatively simple and thus does not completely reflect the non-linear distortion power in the entire range of PA back-off values \cite{lee2014characterization}. Typically maximum transmission power constraint is required. A more suitable option in this matter is a soft-limiter PA model.} 

For a digital M-MIMO system, where each front-end is equipped with a PA and transmits multi-carrier waveform, a similar phenomenon as in the SISO system occurs. If the PA-distorted signal is decomposed into the desired signal and distortion signal, the key issue lies in how the distortion signals from multiple transmitting antennas are added at the reception point. Initially, it was assumed that non-linear distortions from multiple front-ends are uncorrelated, allowing them to be treated as additional white noise emitted from each transmitter antenna \cite{bjornson2014massive}. However, further studies revealed that this is not always the case. The generalized analytical framework for power spectral density (PSD) derivation for M-MIMO, combined with the discussion in \cite{mollen_spatial_char}, show that in some cases, e.g., Line-of-Sight (LoS) channel, the non-linear distortion can be directed towards the intended receiver, preventing the SDR to increase with the number of antennas. Fortunately, as recently shown  in \cite{Salman_OFDM_MIMO_PA_modelin_2023}, this problem vanishes as the channel impulse response prolongs. Therefore, for small-scale Rayleigh fading channel, we assume uncorrelated noise-like distortion from each front-end.   

Until now, there is a limited set of power allocation solutions for M-MIMO systems. 
In \cite{Liu_2021_MMIMO_impairments}, the energy efficiency of an M-MIMO system was maximized considering hardware impairments. However, it was assumed that a narrowband signal is transmitted and that nonlinear distortion power is linearly dependent on the allocated power, similar to \cite{bjornson2014massive}. As discussed in \cite{bjornson2014massive} (Sec. VIIB) the distortion power proportional to the wanted signal power is just a local approximation in a linear range of the PA. These assumptions are far from practical applications of multi-carrier system or full utilization of the PA properties, e.g., operation close to the clipping region.
In \cite{persson2013amplifier}, power allocation per antenna is optimized for a single user equipment (UE) considering the power consumed by PAs. Unfortunately, this does not consider OFDM signal transmission, or signal distortion on a PA. In \cite{Liu_2024_allocation_DMIMO}, power allocation in the downlink (DL) of multiple Distributed MIMO base stations (BSs) in a LoS environment under non-linear PAs is considered. Unfortunately, the study utilizes a particular and limited-order memoryless PA model that requires a total power consumption constraint in the optimization problem definition. Moreover, this study only considers the presence of a single UE in the system. 
\textcolor{blue}{In \cite{hoffmann2023contextual}, transmit power of a M-MIMO system under analog beamforming and soft limiter PA was optimized. In this setup, distortion is directed towards the intended UE, resulting in significant performance degradation. However, while the system serves a single UE at a time, it is not optimizing inter-UE power distribution.} 
Finally, in \cite{Jee_2021_MISO_PA_power_allocation} the authors consider power allocation and precoder design in multi-user DL MIMO transmission. Although the paper provides interesting results, the nonlinearity is modeled only by third-order intermodulation, which is a significant simplification  \cite{lee2014characterization} because higher-order intermodulations are not considered. 
Probably for this reason, the authors had to add a constraint on the total allocated power. Moreover, a single-carrier transmission, though of complex Gaussian signal sample distribution, is considered omitting all difficulties related to a typical scenario of multicarrier MIMO system design.   

In this work, the main contributions are summarized as follows: 
\begin{itemize}
    \item We propose an optimal PA distortion-aware power allocation strategy in a DL M-MIMO OFDM system, utilizing a zero-forcing (ZF) precoder, where the small-scale fading is modeled as a Rayleigh fading channel, and multiple UEs are allocated as multiple layers of a M-MIMO system utilizing all available sub-carriers. We assume that the PA is modeled by a soft-limiter model, which allows for analytical formulation. Most interestingly, the proposed strategy does not need the total transmit power to be constrained because the distortion introduced by the PA self-constrains the result. 
    \item We maximize the sum-rate of the considered system, where we propose an alternating optimization framework to solve the non-convex optimization problem. This enables us to solve two sub-problems, namely, the non-convex total power allocation sub-problem and a convex power distribution among UEs sub-problem. Utilizing the properties of a single variable non-convex sub-problem, we propose an efficient numerical solution for finding a stationary point.
    \item Our proposed solution goes beyond the state-of-the-art in wireless system optimization. In contemporary systems, the in-band distortions are limited at the transmitter, e.g., by setting a limit on Error Vector Magnitude (EVM). Our approach removes this constraint, allowing nonlinear distortion to be introduced at the transmitter as long as the link capacities from the receiver's perspective are optimized. The nonlinear distortion is acceptable, similar to inter-cell interference in cellular systems. The simulation results show that the proposed solution gains the most over legacy systems in the high-range scenario, i.e. when the path loss is substantial relative to the maximum PA emitted power. Therefore, the solution is beneficial for extremely low-power M-MIMO BSs or to support far-away UEs from standard-size and power BSs. For a single 64 and 512-antenna BS serving 60 UEs, a 4 times and 50\% median sum-rate gain over the reference equal power allocation solution is achieved, respectively.
\end{itemize}

The remainder of the paper is organized as follows: first, for the considered system, the wideband signal-to-noise-and-distortion ratio (SNDR) is derived in Section \ref{subsec:sindr}, with the signal model in Section \ref{subsec:sig_mod} and the impact of non-linear PA on the MIMO OFDM system in Section \ref{subsec: PA-NL}. 
Next, the power allocation problem is defined in Section \ref{sec:prob_form} and the alternating optimization solution is presented in Section \ref{sec:prob_sol}. Finally, our simulations in Section \ref{sec:results} show significant sum-rate gains over existing distortion-neglecting solutions. 

\section{System Model}
\label{sec:sys_mod}

This section explains the system model. It shows the influence of non-linear PA on the DL M-MIMO OFDM signal with digital precoding. Although the spatial character of nonlinear distortion significantly depends on the utilized precoder and the wireless channel properties\cite{mollen_spatial_char,Wachowiak2023}, we will focus on the most commonly considered  independent and identically distributed (i.i.d) Rayleigh fading for small-scale fading. 
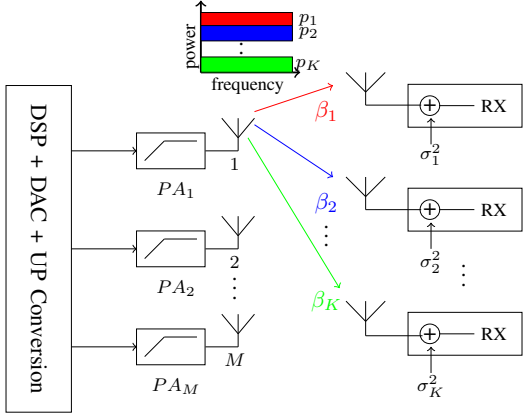
\begin{figure}
 \centering
 \resizebox{0.8\columnwidth}{!}{\input{system_model}}
 \caption{System model: $M$ transmit chains, including common digital signal processing (DSP), digital-to-analog converter (DAC) and up conversion, serving, by ZF precoding, $K$ UEs with $k$-th UE having $\beta_k$ mean channel gain and total (over all antennas) allocated power $p_k$.}
 \label{fig:model}
 \end{figure}  
 
\subsection{Signal Model}
\label{subsec:sig_mod}

As shown in Figure \ref{fig:model},  let us assume that the M-MIMO transmitter is composed of $M$ transmit chains. Each transmit chain utilizes $N$-point Inverse Fast Fourier Transform (IFFT), included in DSP block in Fig. \ref{fig:model}, for modulation. There are $N_{\mathrm{U}}\leq N$ sub-carriers modulated with complex symbols $x_{m,n}$, where $n\in \{1,...,N_{\mathrm{U}}\}$ and $m\in \{1,...,M\}$. \textcolor{blue}{It is assumed that $K$ UEs are simultaneously served  on all sub-carriers using digital precoding. Specifically, the BS transmits a QAM\footnote{\textcolor{blue}{Note that the derived results do not rely on the specific modulation type but on the time-domain statistical behavior of the OFDM signal, and therefore hold for, e.g., QAM, PSK, APSK schemes if OFDM is used with a sufficient number of subcarriers.}} symbol $s_{k,n}$ intended for the $k$-th UE on $n$-th sub-carrier, weighted by the respective precoding coefficient $w_{m,k,n}$. The resulting signal transmitted on $n$-th sub-carrier from $m$-th front-end is}
\begin{equation}
    x_{m,n}=\sum_{k=1}^{K}w_{m,k,n}s_{k,n}.
    \label{eq_precoding}
\end{equation}
Furthermore, we assume that the $k$-th UE is allocated mean power $p_k$ that can be calculated over all antennas and all sub-carriers as
\begin{equation}
\label{eq_UE_power}
    p_k=\sum_{n=1}^{N_{\mathrm{U}}}
    \sum_{m=1}^{M}\mathbb{E}\left[ \left| w_{m,k,n}s_{k,n}\right|^2 \right].
\end{equation}
The precoded signal undergoes inverse Fourier transformation creating the signal sample for the $m$-th front-end at $t$-th time sample ($t\in \{-N_{\mathrm{CP}},...,N-1\}$ as
\begin{equation}
   y_{m,t}= \sum_{n=1}^{N_{\mathrm{U}}} x_{m,n} e^{j2\omega\frac{n}{N}t},
   \label{eq_precoded_signal}
\end{equation}
where $N_{\mathrm{CP}}$ is the length of OFDM cyclic prefix (CP) in samples.
Without loss of generality, a single OFDM symbol is considered. 

\subsection{Influence of non-linear PA on M-MIMO OFDM system}
\label{subsec: PA-NL}
The signals, separately for each front-end, after digital-to-analog conversion and up-conversion to the carrier frequency, pass through the PAs. 
\textcolor{blue}{While there are multiple behavioral models available, not only has the Rapp model been suggested for contemporary Solid State Power Amplifiers \cite{ochiai2013analysis}, but it has also been suggested for the evaluation of $5$G New Radio \cite{Nokia_3gpp_Rapp}.
For the sampled signal $y_{m,t}$, its operation can be modeled as
\begin{equation}
\hat{y}_{m,t}=\frac{y_{m,t}}{\left(1+\frac{|y_{m,t}|^{2p}}{P_{\mathrm{max}}^{p}}\right)^{\frac{1}{2p}}},
\label{eq_Rapp_output}
\end{equation}
where $P_{\mathrm{max}}$ is the saturation power of the power amplifier, which is assumed to be equal for each antenna, though an extension to an unequal case is trivial. Furthermore, $p$ is the smoothness parameter, where $p=2$ is suggested for common PAs \cite{ochiai2013analysis}. Most interestingly, when $p\to \infty$, the Rapp model is equivalent to a soft-limiter PA model, which minimizes non-linear distortion power for OFDM \cite{raich2005_optimal_nonlinearity}. Even if the PA is characterized by some less preferred non-linear characteristics, e.g. $p<<\infty$, it can be assumed that some digital pre-distortion (DPD) has been applied, and as a result the effective characteristic of both elements is close to the soft limiter \cite{Joung_2015_survey_EE_PA}. The soft limiter operation can be modeled as
\begin{equation}
    \hat{y}_{m,t} = \begin{cases} y_{m,t} & \mathrm{for} \left|y_{m,t}\right|^2 \leq P_{\mathrm{max}} \\
                  \sqrt{P_{\mathrm{max}}}e^{j \arg{(y_{m,t})}} & \mathrm{for} \left|y_{m,t}\right|^2 > P_{\mathrm{max}} \end{cases},
                  \label{eq_softlimiter}
\end{equation}
where $\arg(~)$ denotes the argument of a complex number. 
}

While this model has unitary small-signal gain, to focus on non-linear distortion, the linear signal amplification can be \emph{hidden}, e.g., in the precoding coefficients. A typical measure for an operation point of the PA is the Input Back-Off (IBO), which is defined as a ratio of saturation power and the mean input power. \textcolor{blue}{As} we assume i.i.d Rayleigh channel, the precoding should keep equal power distribution between antennas, i.e., $\sum_{k}p_{k}/M$. In such a case, the IBO
\begin{equation}
    \Psi =\frac{P_{\mathrm{max}}}{1/M\sum_{k}p_{k}}.
    \label{eq_IBO_def}
\end{equation}
While the $t$-th sample of the signal $y_{m,t}$ is a sum of $N_{\mathrm{U}}$ active sub-carriers, typically $N_{\mathrm{U}}\gg1$, being modulated with independent random symbols, the Central Limit Theorem can be used to justify that $y_{m,t}$ is complex normal distributed. A more formal justification, even for some special cases like unequal power loading, can be found in \cite{Wei_2010_dist_OFDM}. In such a case, using the Bussgang theorem, the signal at the output of non-linearity can be decomposed as a scaled input and the uncorrelated distortion as
\begin{equation}
\label{eq:bussgang_decomposition}
\hat{y}_{m,t} = \sqrt{\lambda} y_{m,t} + \bar{d}_{m,t}.
\end{equation}
\textcolor{blue}{
Based on the distribution of $y_{m,t}$ amplitudes and the Rapp model in \eqref{eq_Rapp_output}, the signal power scaling factor is defined as 
\cite{Nossek2011power,kryszkiewicz2023efficiency}
\begin{equation}
    \lambda\!=\!\!\left(\!\!\frac{
\mathbb{E} \left[  \hat{y}_{m,t}y_{m,t}^{*} \right]
    }{\mathbb{E} \left[ \left| y_{m,t} \right|^2\right]}\! \right)^2\!\!\!\!\!=
   \left( \int_{0}^{\infty}\frac{2\eta^3}{\left(1+\frac{\eta^{2p}}{\Psi^p} \right)^{\frac{1}{2p}}} e^{-\eta^2} d\eta\right)^2,
    \label{eq_lambda_Rapp}
\end{equation}
where numerical integration is required. However, for the soft-limiter PA an analytical result exist, defining the signal power scaling factor as
\begin{equation}
    \lambda\!=
    \!\left(\!1\!-\!e^{-\Psi}\!\!+\!\!\frac{1}{2}\sqrt{\pi \Psi} \mathrm{\textrm{erfc}}\left(\sqrt{\Psi}\right)\!\right)^2\!,
    \label{eq_lambda}
\end{equation}
with  $\lambda \in (0,1)$.
With the same assumptions, the mean non-linear distortion power at each antenna can be calculated for Rapp PA as \cite{kryszkiewicz2023efficiency}
\footnote{\textcolor{blue}{Observe that the nonlinear distortion model can be easily extended to a non-uniform per antenna power, e.g., as a result of varying path loss over the BS antennas, requiring, $\lambda$, $\Psi$ and $\mathbb{E} \left[ \left| \bar{d}_{m,t} \right|^2\right]$ to become front-end specific.} } 
\begin{equation}
\mathbb{E} \left[ \left| \bar{d}_{m,t} \right|^2\right]=\left(
\int_{0}^{\infty}
    \frac{2 \eta^3}{\left(1+\frac{\eta^{2p}}{\gamma^p} \right)^{\frac{1}{p}}} e^{-\eta^2} d \eta
    -\lambda \right)\sum_{k}\frac{p_{k}}{M},
\label{eq_distortion_SISO_Rapp}
\end{equation}
where again numerical integration is required. In the case of soft-limiter, \eqref{eq_distortion_SISO_Rapp} simplifies to
\begin{equation}
\mathbb{E} \left[ \left| \bar{d}_{m,t} \right|^2\right]=\left(1-e^{-\Psi}-\lambda \right)\sum_{k}\frac{p_{k}}{M}.
\label{eq_distortion_SISO}
\end{equation}
While the Rapp model is not analytically tractable, our derivations focus on soft-limiter PA only for analytically tractability. However, the results for Rapp model are shown in Sec. \ref{subsec:rapp_compare} as well.
}

The signal is transmitted from each antenna reaching $k$-th single antenna UE via $M$ frequency selective channels. After time and frequency synchronization is obtained, the CP is removed and $N$ samples of the received signal undergo FFT processing resulting in a signal on $n$-th sub-carrier of $k$-th UE as
\begin{equation}
\label{rx_sig_freq_domain}
    r_{k,n}=\sum^{M}_{m=1} \mathcal{F}_{[n,t=0,...,N-1]}\{\hat{y}_{m,t}\}h_{m,k,n} + q_{k,n},
\end{equation}
where $\mathcal{F}_{[n,t=0,...,N-1]}$ denotes discrete Fourier transform (DFT) over time samples $0,...,N-1$ at sub-carrier $n$. In addition, $h_{m,k,n}$ denotes channel frequency response on $n$-th sub-carrier for $m$-th antenna and $k$-th UE, and \textcolor{blue}{$q_{k,n}$ is the Additive White Gaussian Noise (AWGN) sample of power $\sigma_{k}^2/N_{\mathrm{U}}$.} Although the channels are assumed to be i.i.d Rayleigh, the large-scale fading can vary between UEs and equals $\beta_k=\mathbb{E}\left[ \left| h_{k,n}\right|^2\right]$. After substitution of (\ref{eq:bussgang_decomposition}) and (\ref{eq_precoded_signal}), the received signal can be decomposed as
\begin{subequations}
\begin{align}
    r_{k,n} &= \sqrt{\lambda} s_{k,n} \sum_{m=1}^{M} h_{m,k,n} w_{m,k,n} \label{rx_sig_freq_domain_decomp1} \\
    &\quad + \sqrt{\lambda} \sum_{\tilde{k}=1, \tilde{k}\neq k}^{K} s_{\tilde{k},n} \sum_{m=1}^{M} h_{m,k,n} w_{m,\tilde{k},n} \label{rx_sig_freq_domain_decomp2} \\
    &\quad + \sum^{M}_{m=1} d_{m,n} h_{m,k,n} \label{rx_sig_freq_domain_decomp3} \\
    &\quad + {\PK q_{k,n}}, \label{rx_sig_freq_domain_decomp4}
\end{align}
\end{subequations}
where $d_{m,n}=\mathcal{F}_{[n,t=0,...,N-1]}\{\bar{d}_{m,t}\}.$

Observe that (\ref{rx_sig_freq_domain_decomp1}) is the signal of interest that undergoes amplification as a result of precoding, described by $\sum_{m=1}^{M}h_{m,k,n}w_{m,k,n}$, and potential attenuation, by a common factor $\sqrt{\lambda}$, caused by the PAs. The component \eqref{rx_sig_freq_domain_decomp2} denotes linear inter-UE interference that is zeroed by employing ZF precoding if the number of antennas is sufficiently large with respect to number of allocated UEs \cite{massivemimobook}.
Finally, (\ref{rx_sig_freq_domain_decomp3}) denotes the non-linear distortion that is a sum over all $M$ transmitting antennas. Based on the Bussgang decomposition, this signal is uncorrelated with $s_{k,n}$. However, the important question is on the directivity of the distortion.   

In \cite{bjornson2014massive}, it is assumed  that the distortion on each front-end is independent white noise, however, this cannot be true for many practical precoders or channels. In \cite{mollen_spatial_char}, the authors derived an analytical formula for calculating the PSD of both the wanted signal and distortion, assuming generalized PA and transmitter structure, which covers both single and multi-carrier transmitter. The numerical examples show that for a single UE and LoS transmission, the distortion can have the same precoding gain as the wanted signal. In this case, an increase in the number of antennas does not improve the SDR. Fortunately, if the number of significant (of high relative gain) channel paths is high\footnote{More specifically, for a single UE, the square of the number of taps has to be at least greater than the number of antennas \cite{mollen_spatial_char}.}, the distortion becomes omnidirectional as well. Since we consider i.i.d Rayleigh fading with independent channel coefficients for every sub-carrier out of $N_{\mathrm{U}}$, this condition is to be met in all practical scenarios. Moreover, the distortion becomes omnidirectional even faster if the number of allocated UEs is higher. As such modeling of distortion as omnidirectionally emitted white noise is allowed in our case. 

\subsection{M-MIMO SINDR definition}
\label{subsec:sindr}

It is assumed that the number of antennas is much larger than the number of UEs allocated in parallel, i.e., $M \gg K$, allowing us to assume that channel hardening occurs {\cite{massivemimobook}}. Therefore, the frequency-selective effect of fading on a particular sub-carrier can be ignored, allowing to represent the wireless channel only by the large-scale fading $\beta_k$. Additionally, it is reasonable to assume that all UEs utilize all the available $N_{\mathrm{U}}$ sub-carriers.

If we do not consider non-linear PA, the coefficient $\lambda$ becomes one, and the component (\ref{rx_sig_freq_domain_decomp3}) vanishes. For ZF precoding, the component (\ref{rx_sig_freq_domain_decomp2}) vanishes as well, resulting in the $k$-th UE Signal-to-Interference and Noise power Ratio (SINR) formula\footnote{With Maximal Ratio Transmission (MRT) precoding, \eqref{eq: SINR_perfect_ZF} changes to 
\begin{equation}
    \tilde{\gamma}^{\mathrm{MRT, perfect PA}}_k =  \frac{M\cdot p_{k}\cdot \beta_{k}}{\sigma_k^2 + \beta_{k} \sum_{\tilde{k}\neq k}p_{\tilde{k}} }. 
\label{eq: SINR_perfect_MRT}
\nonumber
\end{equation}
} \cite{massivemimobook} 
\begin{equation}
    \tilde{\gamma}^{\mathrm{ZF,perfect PA}}_k =  \frac{(M - K) \cdot p_{k}\cdot \beta_{k}}{\sigma_k^2}, 
\label{eq: SINR_perfect_ZF}
\end{equation}
where $\sigma_k^2$ is the AWGN noise power over all $N_{\mathrm{U}}$ sub-carriers and recall that $p_k$ is defined in \eqref{eq_UE_power}.

While soft-limiter PA characteristic is considered, the following lemma can be used to calculate Signal to Interference, Noise and Distortion power Ratio (SINDR):
\begin{theorem}
The $k$-th UE SINDR under soft-limiter PA and ZF precoding can be calculated as
\footnote{If MRT precoding is considered \eqref{eq: SINR3} changes to 
\begin{equation}
    \gamma_k^{\mathrm{MRT}} =  \frac{M\cdot \lambda \cdot p_{k}\cdot \beta_{k}}{\sigma_k^2 + \beta_{k}D+\beta_{k}\lambda \sum_{\tilde{k}\neq k}p_{\tilde{k}} }.
    \nonumber
\label{eq: SINR_MRT}
\end{equation}
}
\begin{equation}
    \gamma_k^{\mathrm{ZF}} =  \frac{(M - K)\cdot \lambda \cdot p_{k}\cdot \beta_{k}}{\sigma_k^2 + \beta_{k}D }.
\label{eq: SINR3}
\end{equation}
\end{theorem}
\begin{proof}
    Observe that if the non-linear PA is considered, the wanted signal and linear inter-UE interference is simply scaled by a factor $\sqrt{\lambda}$ as shown in (\ref{rx_sig_freq_domain_decomp1}) and (\ref{rx_sig_freq_domain_decomp2}), respectively. As such the power is scaled by $\lambda$. Furthermore, the non-linear distortion can be assumed to be omnidirectional, as explained in Section \ref{subsec: PA-NL}, allowing us to model it as Gaussian noise of power defined for each transmitting antenna by (\ref{eq_distortion_SISO}). However, it is important to understand that PA-based distortion is not frequency flat and has bandwidth wider than the wanted signal. As derived in \cite{lee2014characterization}, the total distortion power emitted in band of the transmission is no larger than $\eta=2/3$ of the total distortion power. The rest of the interference can be emitted in the adjacent band or attenuated by a spectrum shaping filter at the front-end, as required by a given communications system's spectrum emission mask. Therefore, the in-band distortion power at $m$-th front-end can be approximated as
\begin{equation}
    \mathbb{E}\left[\left|d_{m,n}\right|^2\right]=\eta \mathbb{E} \left[ \left| \bar{d}_{m,t} \right|^2\right].
\end{equation}
While there are $M$ sources of distortion whose signal add in an uncorrelated manner at $k$-th UE after transmission through channels of mean gain $\beta_k$, the total received distortion power is
\begin{equation}
   D_{k}=M \beta_k \eta \mathbb{E} \left[ \left| \bar{d}_{m,t} \right|^2\right].
   \label{eq_per_UE_distortion}
\end{equation}
Observe that the distortion power at transmitter antenna $m$ $\mathbb{E} \left[ \left| \bar{d}_{m,t} \right|^2\right]$ is independent from each individual UE channel gain $\beta_k$. This allows us to define the effective distortion power at the transmitter by combining (\ref{eq_per_UE_distortion}) and (\ref{eq_distortion_SISO}) while omitting $\beta_k$:
 \begin{equation}
   D= \eta \left(1-e^{-\Psi}-\lambda \right)\sum_{k}p_{k}. 
   \label{eq_per_UE_distortion_effective}
\end{equation}
With this definition, the SINDR for the ZF precoding\footnote{\textcolor{blue}{Here, the focus is on the ZF precoding as this configuration can gain the most from the distortion awareness, in comparison to the MRT case. The inter-UE linear interference can be much stronger than the nonlinear distortion when MRT precoding is employed, reducing the potential throughput gains.}} under non-linear PA can be computed as in (\ref{eq: SINR3}).
\end{proof}
Observe that this formula shows that both the desired signal and the distortion scale with the transmit power $p_k$ and $\sum_{k}p_{k}$, respectively, as in \cite{bjornson2014massive}. However, here the proportionality factors $\lambda$ and $\eta \left(1-e^{-\Psi}-\lambda \right)$ scale with the IBO $\Psi$, constituting a more general model in comparison to \cite{bjornson2014massive}, where constant proportionality coefficients are used. The authors in \cite{bjornson2014massive} (Sec. VII B) were aware that their model was limited only to the linear dynamic range of the amplifier or locally, when a very small variation of IBO is considered. 
\begin{remark}
     
The equations above use some approximations, i.e. assumption of channel hardening, and the upper bound of the in-band distortion power. The link level simulations have been carried out to show the accuracy of approximations. The SDR at the receiver $\gamma^{\mathrm{ZF}}_k$ for $\sigma_k^2=0$, was estimated while varying the IBO value for a various number of utilized sub-carriers, allocated UEs, and two considered precoders. The simulations have been carried out for the OFDM transmitter, utilizing $N=512$ sub-carriers with $N_{\mathrm{U}}=100$ sub-carriers modulated with random 16 PSK symbols. The power was equally divided among the UEs, i.e. $\forall_{k\neq \tilde{k}} p_k=p_{\tilde{k}}$. As expected, it can be observed from Fig. \ref{fig:SDR_plot} that
the higher the IBO (ergo: higher clipping power $P_{\mathrm{max}}$) the higher the SDR. Moreover, for a given number of allocated UEs, an increase in the number of antennas improves the SDR as a result of a higher precoding gain and higher wanted signal power. Most importantly, in all tested cases the analytical formula follows the simulation-based SDR result with an error no greater than 0.7 dB. 

\begin{figure}[!t]
\centering
\includegraphics[width=\columnwidth]{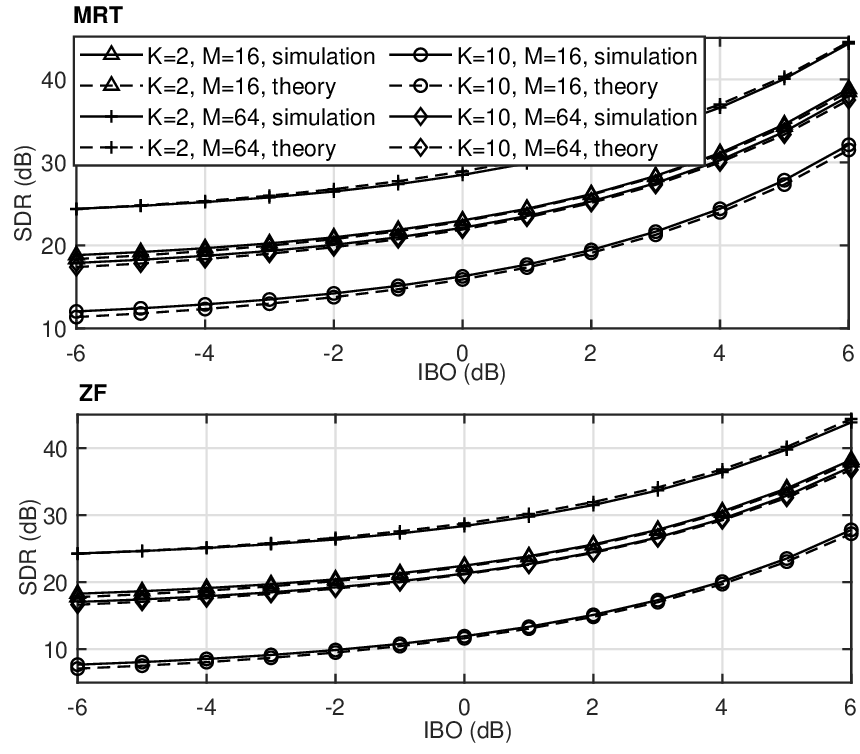}
\caption{In-band SDR for MRT and ZF precoding and various number of antennas (M) and UEs (K) obtained by simulation and analytical formulas.}
\label{fig:SDR_plot}
\end{figure}   
\end{remark}

In the following, we are interested in observing the impact of PA non-linearity on the sum-rate and therefore, we formulate the sum-rate maximization problem.  

\section{Problem Formulation}
\label{sec:prob_form}

In M-MIMO systems employing ZF precoding, the achievable rate for $k$-th UE in the DL, can be modeled with SINDR $\gamma_k^{ZF}$ from \eqref{eq: SINR3} as  
\begin{equation}
    R_k = B \log_2 (1 + \gamma_k^{ZF}), 
\label{eq: datarate2}
\end{equation}
where, the allocated bandwidth $B$ is the product of the number of allocated sub-carriers and the bandwidth per sub-carrier. 
The large number of antennas at the BS results in the phenomenon of channel hardening,  mitigating the effects of small-scale fading and reducing the channel variations \cite{massivemimobook}. Asymptotically, the variances of the channel gains vanish completely. This allows to use the average channel gain $\beta_k$ for choosing the modulation and coding scheme. The maximum achievable rate is then given by the expression in (\ref{eq: datarate2}). 
With this expression, the sum-rate maximization problem, where $\boldsymbol{p}=[p_1,...,p_K]$, can then be formulated as
\begin{alignat}{2}
\max_{\boldsymbol{p} \geq \boldsymbol{0}}   &\quad&     & B \sum_k R_k. 
\label{eq:optprob}
\end{alignat}

\begin{prop}\label{prop1}
    The problem in \eqref{eq:optprob} is non-convex. 
\end{prop}

\begin{proof}
The proof, based on the analysis of the eigenvalues of the Hessian matrix of the sum-rate function, is given in Appendix \ref{sec: non-convexity}. 
\end{proof}
Due to the non-convexity of \eqref{eq:optprob}, solving for individual power allocation $p_k$ is non-trivial. Therefore, by defining $p_k = \omega_k P$, we deconstruct the problem in two sub-problems, namely, Distortion Aware total Power $(P)$ Allocation  (DAPA) and Fixed Power  Distribution Algorithm $(\omega_k)$ (FPDA). While $P$ is the sum allocated power for all PA that is to support data transmission to all $K$ UEs, the factor $\omega_k$ denotes which part of the power $P$ is allocated to UE $k$, requiring that $\omega_k\geq 0$ and $\sum_k \omega_k =1$ as $\sum_k p_k=P \sum_k \omega_k =P$.

Using the substitution above, the optimization problem in \eqref{eq:optprob} can be formulated as 
\begin{subequations}
    \begin{alignat}{2}
    \max_{P \geq 0, \boldsymbol{\omega} \geq \boldsymbol{0}} &\quad& & \sum_k B \log_2\left(1 + \frac{(M - K) \cdot \lambda \cdot \omega_{k} P \cdot \beta_{k}}{\sigma_k^2 + \beta_{k} D} \right) \label{eq:obj1} \\
    \text{subject to} &\quad& & \sum_k \omega_{k} = 1.\label{eq:sum_const} 
    \end{alignat}
    \label{eq:optprob3}
\end{subequations}
The equivalent problem in \eqref{eq:optprob3} enables us to decouple the joint optimization over $P$ and $\boldsymbol{\omega}$ into two optimization problems: first over a scalar variable P and second over a vector $\boldsymbol{\omega}$. Thus, we formulate an alternating optimization problem, where keeping power per UE fixed, i.e. for a fixed $\boldsymbol{\omega} = [\omega_1,\dots, \omega_k]$, we solve
\begin{equation}
    \max_{P \geq 0}  \sum_k B\log_2\left(1 + \frac{(M - K) \cdot \lambda \cdot \omega_{k} P \cdot \beta_{k}}{\sigma_k^2 + \beta_{k} D} \right).
    \label{eq:TPA} 
\end{equation}
We call this a DAPA problem. Subsequently,
for a fixed $P$ we solve 
\begin{subequations}
    \begin{alignat}{2}
    \!\!\!\!\max_{\boldsymbol{\omega} \geq \boldsymbol{0}} &\ \ & & \sum_k B \log_2\left(1 + \frac{(M - K) \cdot \lambda \cdot \omega_{k} P \cdot \beta_{k}}{\sigma_k^2 + \beta_{k} D} \right) \label{eq:obj3} \\
    \!\!\!\!\text{subject to} &\ \ & & \sum_k \omega_{k} = 1,\label{eq:pc3} 
    \end{alignat}
    \label{eq:optprob5}
\end{subequations}
 and call it FPDA problem.

\begin{prop} \label{prop:equivalence}
The solution to \eqref{eq:TPA} and \eqref{eq:optprob5} achieve a stationary point of \eqref{eq:optprob3},  which corresponds to finding the stationary point of \eqref{eq:optprob}.
\end{prop}

\begin{proof}
    The proof is given in Appendix \ref{sec:equivalence}.
\end{proof}

\section {Solution of the optimization problem}
\label{sec:prob_sol}

In this section, we describe the solution to the proposed DAPA and FPDA sub-problems. It should be noted that the solutions to \eqref{eq:TPA} and \eqref{eq:optprob5} are used in an alternating manner
as will be shown in Sec \ref{sec_alternating}. 

\subsection{Distortion Aware total Power Allocation (DAPA)}
\label{sec_total_alg}
The solution to the non-convex DAPA problem equates to finding the stationary point of \eqref{eq:TPA}. This can be achieved by solving
\begin{equation}    \at{\sum_k\frac{\partial  \tilde{R}_k}{\partial P} }{P=\tilde{P}}=0
    \label{eq_sum_d_Rk}
\end{equation}
for the root $\tilde{P}$, where
\begin{equation}
    \tilde{R}_k=B\log_2\left(1 + \frac{(M - K) \cdot \lambda \cdot \omega_{k} P \cdot \beta_{k}}{\sigma_k^2 + \beta_{k} D} \right).
\end{equation}

In order to find $\tilde{P}$ first, properties of $\frac{\partial  \tilde{R}_k}{\partial P}$ can be analyzed using the following lemma: 
\begin{theorem}\label{lem:DAPA_1}
The function $\frac{\partial \tilde{R}_k}{\partial P}$ has a single root equal to the root of the function
\begin{equation}
f_{k}(P)=\frac{2\sigma_k^2}{\sqrt{\pi}\beta_k \eta M P_{max}}-\frac{\textrm{erfc}(\sqrt{\Psi})}{\sqrt{\Psi}}.
\label{eq_f_k}
\end{equation}
The function is monotonically decreasing in the entire range of $P$, i.e., $P \in [ 0, \infty )$.
\end{theorem}

\begin{proof}
The proof of the Lemma \ref{lem:DAPA_1} is provided in Appendix \ref{sec:proof_lemma_1}, where the derivative $\frac{\partial  \tilde{R}_k}{\partial P}$ is computed and its characteristic is analyzed. 
\end{proof}

Through Lemma \ref{lem:DAPA_1} we show that $\frac{\partial  \tilde{R}_k}{\partial P}$ has a single root, but a way of finding it is another issue. While unaware of an analytical solution, the following upper and lower bounds can be defined allowing for efficient numerical solutions. 
 \begin{theorem}\label{lemma_upper_lower_root}
     The root $\tilde{P}_k$ of function $f_k(P)$ belongs to the range $[\tilde{P}^{lower}_k, \tilde{P}^{upper}_k ]$, where     
     \begin{equation}
   \tilde{P}^{lower}_{k}=\frac{2MP_{max}}{W\left(2\left(\frac{\sqrt{\pi}\beta_k \eta M P_{max}}{2\sigma_k^2}\right)^2 \right)},
   \label{eq_root_P_lower}
\end{equation}
\begin{equation}
       \tilde{P}^{upper}_{k}=\frac{4MP_{max}}{W\left(4\left(\frac{\sqrt{e}\beta_k \eta M P_{max}}{2\sqrt{2}\sigma_k^2}\right)^2 \right)}
       \label{eq_root_P_upper}
\end{equation}
  and $W(~)$ returns the principal branch of the Lambert W function. 
 \end{theorem}
 \begin{proof}
 The proof of Lemma \ref{lemma_upper_lower_root} is sketched in detail in Appendix \ref{sec:DAPA_proof_2}, where the upper and lower bounds of the $\textrm{erfc}(x)$ function are used to find the upper and lower bounds on $\tilde{P}_k$, enabling us to apply numerical methods, such as a bisection method, to find the solution.      
\end{proof}
Finally we need to find the root $\tilde{P}$ of $\sum_k\frac{\partial  \tilde{R}_k}{\partial P}$ as defined in \eqref{eq_sum_d_Rk}. Therefore, we first analyze the roots of all (over $k$) functions $f_{k}(P)$. Observe that the root $\tilde{P}_{k_{min}}: k_{min}=\argmin_k \frac{\sigma_k^2}{\beta_k}$ of the function $f_{k_{min}}(P)$ will be minimal over all roots of functions $f_{k}(P)$. This can be justified by the structure of function \eqref{eq_f_k} decreasing from the initial value $\frac{2\sigma_k^2}{\sqrt{\pi}\beta_k \eta M P_{max}}$ for $P=0$ (see Lemma \ref{lem:DAPA_1} and its proof). While the decrease rate is invariant of $k$ (see Appendix \ref{sec:proof_lemma_1}, \eqref{eq_f_k_prim}), the lower the $\frac{\sigma_k^2}{\beta_k}$, the lower the root of $\frac{\partial R_k}{\partial P}$. Similarly, the maximum root $\tilde{P}_{k_{max}}$ will be obtained for $k_{max}=\argmax_k \frac{\sigma_k^2}{\beta_k}$. 

While each component function $\frac{\partial  \tilde{R}_k}{\partial P}$ is non-negative for $P<\tilde{P}_k$ and non-positive for $P>\tilde{P}_k$ as discussed below (\ref{eq_lim_f_k_prim}), their sum will be non-negative for $P<\tilde{P}_{k_{min}}$ and non-positive for $P>\tilde{P}_{k_{max}}$. The root $\tilde{P}$ can be found in the range  $[\tilde{P}_{k_{min}}, \tilde{P}_{k_{max}}]$, e.g., numerically. 

However, the problem is in finding $\tilde{P}_{k_{min}}$ and $\tilde{P}_{k_{max}}$. For this purpose Lemma \ref{lemma_upper_lower_root} can be used. Observe that $\tilde{P}^{lower}_{k_{min}}\leq\tilde{P}_{k_{min}}$, and $\tilde{P}^{upper}_{k_{max}}\geq\tilde{P}_{k_{max}}$. Therefore, the root $\tilde{P}$ belongs to the range $[ \tilde{P}^{lower}_{k_{min}}, \tilde{P}^{upper}_{k_{max}}]$. While $\at{\sum_k\frac{\partial  \tilde{R}_k}{\partial P} }{P=\tilde{P}^{lower}_{k_{min}}}$ is non-negative and $\at{\sum_k\frac{\partial  \tilde{R}_k}{\partial P} }{P=\tilde{P}^{lower}_{k_{max}}}$ in non-positive, numerical methods of finding a root in this range can be used.  
Lemma \ref{lem:DAPA_1} and \ref{lemma_upper_lower_root} allows us to propose Algorithm \ref{alg:subp1} which first finds  $\tilde{P}^{lower}_{k_{min}}$ and $\tilde{P}^{upper}_{k_{max}}$ and then employs bisection method to find the root of $\sum_k\frac{\partial  \tilde{R}_k}{\partial P}$ in the above-mentioned range. 

\begin{algorithm}
\caption{Distortion Aware total Power Allocation (DAPA)}\label{alg:subp1}
\begin{algorithmic}[1]
\State $k_{min} \gets \argmin_k \frac{\sigma_k^2}{\beta_k}$ ; $k_{max} \gets \argmax_k \frac{\sigma_k^2}{\beta_k}$
\State $P_L \gets \tilde{P}^{lower}_{k_{min}}$ using (\ref{eq_root_P_lower}) ; $P_R \gets \tilde{P}^{upper}_{k_{max}}$ using (\ref{eq_root_P_upper})
\While{$P_R-P_L>\delta$ }
\State $P_C \gets 0.5P_L+0.5P_R$
\If{$\at{\frac{\partial \sum \tilde{R}_{k}}{\partial P}}{P=P_C}>0$ using (\ref{eq_derivative_Rk})}
    \State $P_L \gets P_C$
\Else
    \State $P_R \gets P_C$
\EndIf
\EndWhile
\end{algorithmic}
\end{algorithm}

It should be noted that the accuracy of the output result $P_{C}$ is limited by a non-negative value $\delta$ used in Line 3 of Algorithm ~\ref{alg:subp1}, i.e., the root belongs to the range $[P_{C}-\frac{\delta}{2}, P_{C}+\frac{\delta}{2}] $. 

\subsection{Fixed Power Distribution Algorithm (FPDA)}
\label{sec_distr_alg}
The FPDA aims at optimally distributing the allocated total power $P$ among the UEs, which is obtained by solving \eqref{eq:TPA}. For a fixed $P$, we apply the Karush-Kuhn-Tucker (KKT) conditions to achieve optimal total power distribution and propose the following:

\begin{theorem} \label{lem:DAPAD}
    For a fixed P, the optimal solution of the optimization problem in \eqref{eq:optprob5}, represented by $\omega^*_k$, is
    \begin{align}
        \omega^*_k=\max \left\{ 0 , \frac{B}{\nu^*} - \frac{\sigma^2_k + \beta_kD}{(M-K)\lambda P\beta_k} \right\},
    \label{eq:opt_omega}
    \end{align}
    where $\nu^*$ is the optimal solution of Lagrange dual problem.
\end{theorem}

\begin{proof}
    The proof is sketched in detail in Appendix \ref{sec:DAPAD_proof}. We formulate the Lagrange dual function and apply the KKT conditions to find the optimal solution. 
\end{proof}
It can be observed that the solution to $\omega_k^*$ in \eqref{eq:opt_omega} resembles the well known water filling solution, which is a piecewise-linear increasing function of $\frac{B}{\nu^*}$, with break points at $\frac{\sigma^2_k + \beta_kD}{(M-K)\lambda P\beta_k}$ for each $k$. Therefore, the optimum $\nu^*$ can be found either by the well known bisection method or a single loop solution provided in Algorithm \ref{algo:wf}.  

\begin{algorithm}
    \caption{Fixed Power Distribution Algorithm (FPDA)
    } \label{algo:wf}
    \begin{algorithmic}[1]
    \State \text{Input}: \text{$P$ from \textbf{Algorithm \ref{alg:subp1}}}
    
    \State \text{Compute}: $\boldsymbol{G} = \{G_k \gets \frac{\sigma^2_k + \beta_kD}{(M-K)\lambda P\beta_k}\: |\: \forall_{k>\tilde{k}} G_k\geq G_{\tilde{k}} \} $ 
    
        \State Initialize: $\boldsymbol{\omega} \gets \boldsymbol{0}$, $\omega_{\text{rem}} \gets 1$      
        \For{$k = 0 \to K - 1$}
            \If{$k < K - 1$}
                \If{$\omega_{\text{rem}} \leq k \cdot |\boldsymbol{G}_{k + 1} - \boldsymbol{G}_k|$}
                     \For{$m = 0 \to k+1$}
                     \State ${\omega}_{m} \gets {\omega}_{m} + \frac{\omega_{\text{rem}}}{k + 1}$
                    \EndFor 
                    \State \textbf{break}
                \Else
                 \For{$m = 0 \to k+1$}
                     \State ${\omega}_{m} \gets {\omega}_{m} + |\boldsymbol{G}_{k + 1} - \boldsymbol{G}_k| $
                    \EndFor
                    \State $\omega_{\text{rem}} \gets \omega_{\text{rem}} - k \cdot 
                    |\boldsymbol{G}_{k + 1} - \boldsymbol{G}_k|$
                \EndIf
            \Else
            \For{$m = 0 \to k+1$}
                     \State ${\omega}_{m} \gets {\omega}_{m} + \frac{\omega_{\text{rem}}}{k + 1} $
                    \EndFor
            \EndIf
        \EndFor
    \end{algorithmic}
\end{algorithm}

\subsection{Alternating Optimization}
  \label{sec_alternating}
The solutions of sub-problems presented in Sec. \ref{sec_total_alg} and Sec. \ref{sec_distr_alg} can be used to solve the original optimization problem (\ref{eq:optprob3}), and thereby \eqref{eq:optprob}, by  alternating optimization as shown in Algorithm \ref{alg:ao}. Firstly, for a fixed $\boldsymbol{\omega} = [\omega_1,\dots, \omega_k]$, we find the optimum solution for \eqref{eq:TPA} via Algorithm \ref{alg:subp1} and then with this solution, i.e. for a fixed $P$, we find the optimum solution for \eqref{eq:optprob5} via Algorithm \ref{algo:wf}. This process is repeated over multiple iterations of $P$ and $\omega$, with superscript $(i)$ denoting the solution at the $i$-th iteration, until convergence is achieved, i.e, until the difference between total allocated power in the current and previous iteration is smaller than a non-negative value $\delta$, which measures the accuracy of the root.
\begin{algorithm}
\caption{Alternating Optimization}\label{alg:ao}
\begin{algorithmic}[1]
\State Initialize: $i \gets 0$,  $\boldsymbol{\omega^{(i)}} = [\omega_1^{(i)}, \dots, \omega_K^{(i)}] \gets \frac{1}{K}$
\Repeat
\State $i \gets i+1$
\State solve \eqref{eq:TPA} given $\boldsymbol{\omega^{(i-1)}}$ to get $P^{(i)}$ by \textbf{ Algorithm \ref{alg:subp1}}
\State solve \eqref{eq:optprob5} given $P^{(i)}$ to get $\boldsymbol{\omega^{(i)}}$ by \textbf{ Algorithm \ref{algo:wf}}
\Until${|P^{(i-1)} - P^{(i)}| < \delta}$
\end{algorithmic}
\end{algorithm}

Note that the algorithm can be easily extended to account for the sum-rate in the convergence criterion.

\section{Numerical Results}
\label{sec:results}

In this section, we evaluate the performance of the proposed DAPA-FPDA algorithm under different scenarios and compare the results with certain benchmark algorithms, a summary of which is provided in Table \ref{tab: algos}. 
\begin{table}[!ht]
    \centering
    \begin{adjustbox}{width=1.0\columnwidth}
        \begin{tabular}{|>{\columncolor[HTML]{9B9B9B}}l|l|}
        \hline
        DAPA-FPDA & Distortion Aware total Power Allocation using Algorithm \ref{alg:subp1}with \\ & Fixed Power Distribution Algorithm using Algorithm \ref{algo:wf}   \\ \hline
        DAPA-E & Distortion Aware total Power Allocation using Algorithm \ref{alg:subp1} with \\  & Equal per UE power distribution \\ \hline
        REF-E & Reference total power (IBO equals 6 dB) with \\ & Equal per UE power distribution  \\ \hline
        REF-FPDA & Fixed IBO based Total Power allocation with \\ & Fixed Power Distribution Algorithm using Algorithm \ref{algo:wf}\\ \hline     
        \end{tabular}
    \end{adjustbox}
\caption{Tested algorithms abbreviation summary.}
\label{tab: algos}
\end{table}
\begin{table}[!ht]
\renewcommand{\arraystretch}{1.2}
    \centering
    \begin{adjustbox}{width=1.0\columnwidth}
        \begin{tabular}{|c|c|c|c|}
        \hline
        \rowcolor[HTML]{9B9B9B}
         $N_{\mathrm{U}}$ & $\Delta f$ & $\eta$ & $(\sigma_k^2)_{dBm}$  \\ \hline
        $1200$ & $15$ kHz & $\frac{2}{3}$ & $-174 \frac{dBm}{Hz} + 10  \cdot \log_{10}(N_{\mathrm{U}} \cdot \Delta f)$ \\ \hline
        \rowcolor[HTML]{9B9B9B}
        $M$ & $K$ & $P_{\textrm{max}}$ & $(\beta_k)_{\mathrm{dB}}$ \cite{itu-m2135-1-2009}  \\ \hline
        $64, 512$ & $2, 20, 60$ & $100$ mW & $22.7 + 36.7 \cdot \log_{10}(d_k) + 26 \cdot \log_{10}(f_c)$  \\ \hline     
        \end{tabular}
    \end{adjustbox}
\caption{Simulation parameters, where $d_k$ is the distance between UE $k$ and the base station in meters and $f_c =  3$ GHz \cite{itu-m2135-1-2009}.}
\label{tab: sys_para}
\end{table}
It should be noted that $\Psi = 6$ dB has been used for the reference algorithms REF-E and REF-FPDA because the SDR at this value for a single transmitter equals around $27$ dB. This is equivalent to the mean EVM of $4.5\%$ required in $5$G New Radio for $256$-QAM constellation \cite{3gpp_38141}. Additionally, Table \ref{tab: sys_para} shows the simulation parameters.

Please keep in mind that in this section path loss in dB is denoted as $(\beta)_{\mathrm{dB}}$ and is related to the linear channel gain $\beta$ used in the previous sections as $\left( \beta \right)_{\mathrm{dB}}=-10\log_{10}(\beta)$.

\subsection{Multi-UE Homogeneous Path Loss}
\label{subsec:multi_use_same_pathloss}
\begin{figure*}[!t]
    \centering
    \captionsetup[subfloat]{labelfont=small,textfont=small}
    \subfloat[sum-rate over varying path loss.  \label{fig:ratio_comp_DAPA_FPDA_REFE}]{
        \includegraphics[width=0.95\columnwidth]{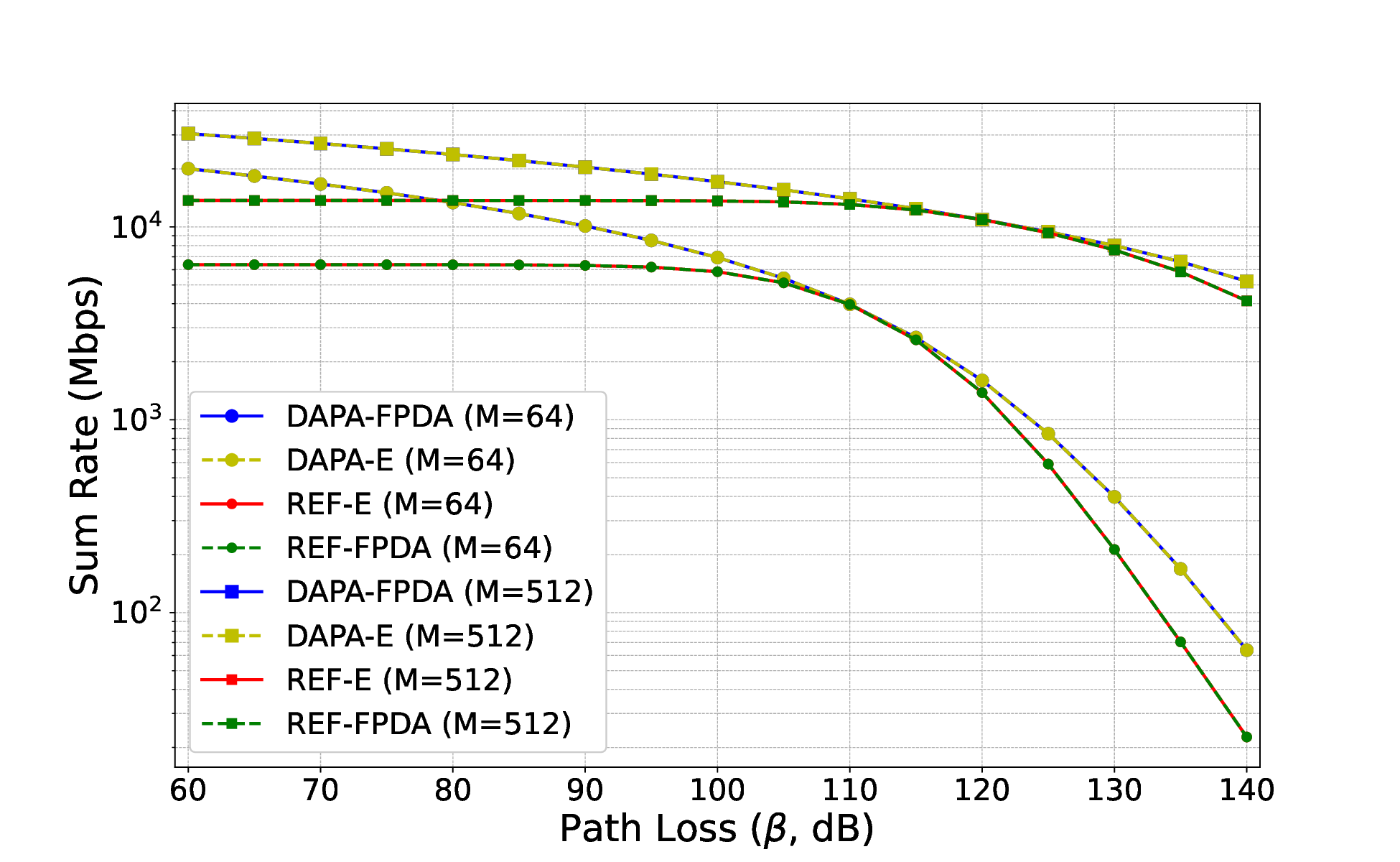} 
    }
    \subfloat[IBO over varying path loss. \label{fig:ratio_comp_DAPA_FPDA_DAPAE}]{
        \includegraphics[width=0.95\columnwidth]{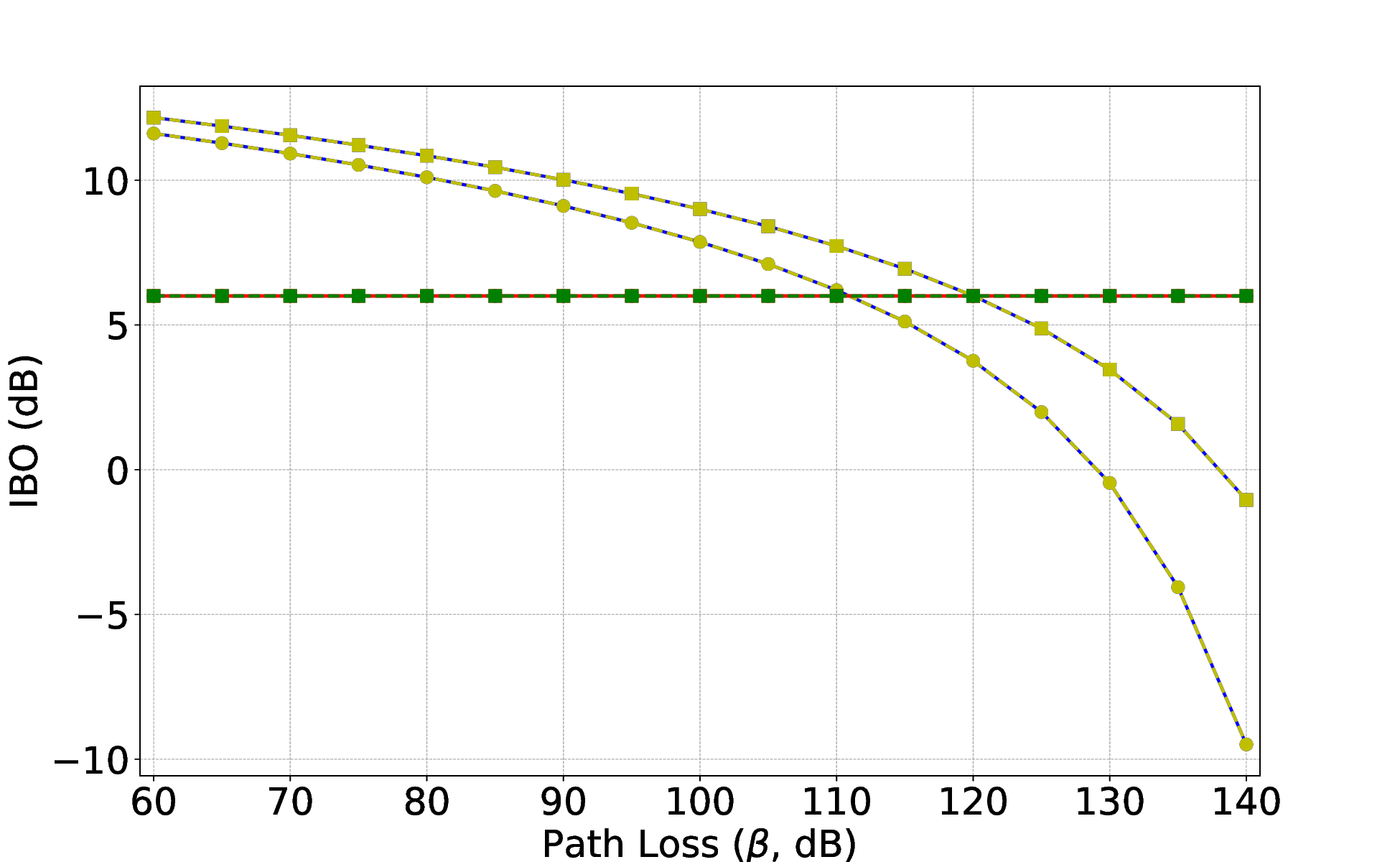} 
    }
    
    \caption{sum-rate and IBO comparison for $K=60$ of equal path loss, i.e., symmetric UEs.}
    \label{fig:rate_ibo}
\end{figure*}

With $K = 20$ UEs, $M = 64$ or $M = 512$ antennas, and $P_{\mathrm{max}} = 100$ mW, Fig. \ref{fig:rate_ibo} shows the changes in the sum-rate and the IBO over different path loss values, where all the UEs have the same path loss, which, e.g., is possible when the UEs are at the same distance from BS but on various azimuth angles.
Firstly, it can be observed that, as expected, the sum-rate decreases as the channel worsens. More interestingly, if IBO adaptation is allowed, i.e. DAPA-FPDA and DAPA-E, its value also decreases. This is a similar effect as in a single-UE, SISO OFDM system \cite{kryszkiewicz2023efficiency}. An optimal IBO decreasing with path loss shows that the further away the receiver is, a higher power of nonlinear distortion is acceptable (also measurable as decreasing SDR value). At the same time the SNR increases as the wanted signal power increases. 
DAPA-FPDA and DAPA-E achieve the same performance as all UEs have the same path loss, resulting in equal power allocation as the optimal strategy. 
Most importantly, DAPA-FPDA and DAPA-E outperform reference solutions (REF-E and REF-FPDA) in terms of sum-rate, both when path loss is greater and smaller than around $100$ dB for $M=64$ antennas. For $100$ dB path loss, an IBO of $6$ dB is optimal, where all algorithms achieve an equal sum-rate. However, observe that as the number of antennas increases from $64$ to $512$, the IBO crossing point shifts to around $110$ dB path loss. This can be justified by looking at (\ref{eq_f_k}). Observe that in the considered case $f_k(P)$ is equal for all UEs. As this function increases monotonically in IBO, a higher value of $M$ leads to a lower value of $\frac{2\sigma_k^2}{\sqrt{\pi}\beta_k \eta M P_{max}}$ and therefore the root will be observed for higher IBO values. Most interestingly, observe that for the homogeneous UE case, the optimal IBO is independent of the number of allocated UEs, i.e., there is no $K$ in (\ref{eq_f_k}). Less formally, the increase in optimal IBO as a function of the number of antennas can be justified by the increased SNR at a given distance with a higher number of antennas, i.e., precoding gain $(M-K)$. Since the SNR is improved the nonlinear distortion power should be reduced, which can be achieved by increasing the IBO. As expected with a higher number of antennas a higher throughput is achievable.

\subsection{Two-UE Heterogeneous Path Loss}
\label{subsec:two_ue_het_loss}

In this section $K = 2$ UEs are considered for $M = 64$ antennas, and $P_{\mathrm{max}} = 100$ mW. All possible UEs path loss combinations on a rectangular grid of values from 60 dB to 150 dB with 1 dB steps have been tested. First, Fig.  \ref{fig:ratio_DAPA_FPDA_REFE}
shows the ratio of sum-rate achievable by the proposed DAPA-FPDA vs REF-E method. On the diagonal, where $\beta_1=\beta_2$, an improved rate is achievable when UEs are close to the BS, i.e., path loss below 90 dB, or far away, i.e., path loss above 120 dB. The achievable rate can be more than 2 times higher than in the reference scenario. This results from the adjustment of the IBO value that is visible on diagonal of Fig. \ref{fig:IBOopt}, showing IBO value for DAPA-FPDA. The highest gain is observed for a large path loss for which setting low IBO decreases SDR but as the system is noise-limited at this distance it allows to increase SNDR. For the path loss around 100 dB, the optimal IBO value equals 6 dB, which is the same as the one used in REF-E scenario. This is in line with the results of Section \ref{subsec:multi_use_same_pathloss}. As expected, in the entire range of path loss values, as long as $\beta_1=\beta_2$, both UEs are allocated equal power, i.e., $\omega_1=\omega_2=0.5$, which is also visible in Fig. \ref{fig:omega_distribution_DAPA_FPDA}. 
\begin{figure}[!t]
 \centering
  \includegraphics[width=0.8\columnwidth]
  {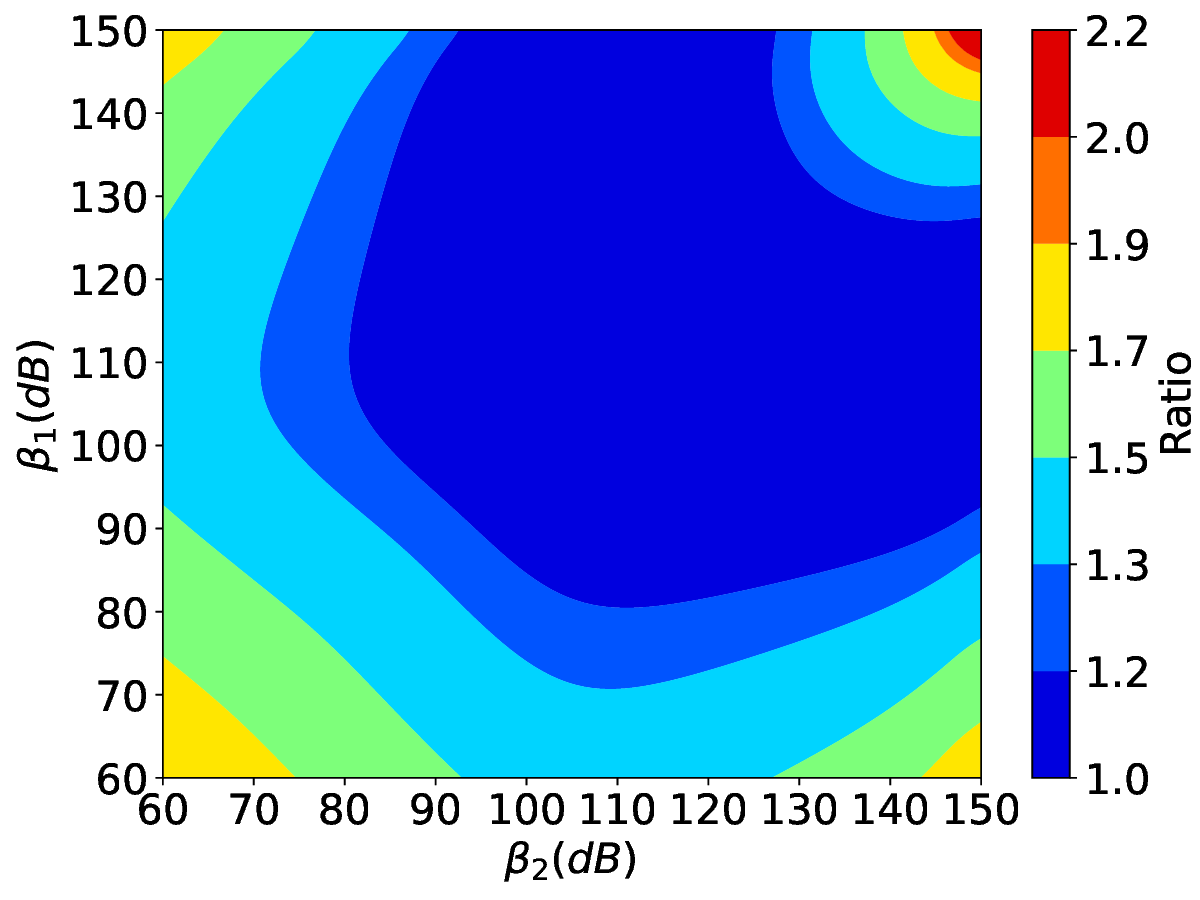}
 \caption{Ratio of DAPA-FPDA and REF-E sum-rate for 2 UE case.}
  \label{fig:ratio_DAPA_FPDA_REFE}
 \end{figure}   

\begin{figure}[!t]
\centering
\includegraphics[width=0.8\columnwidth]{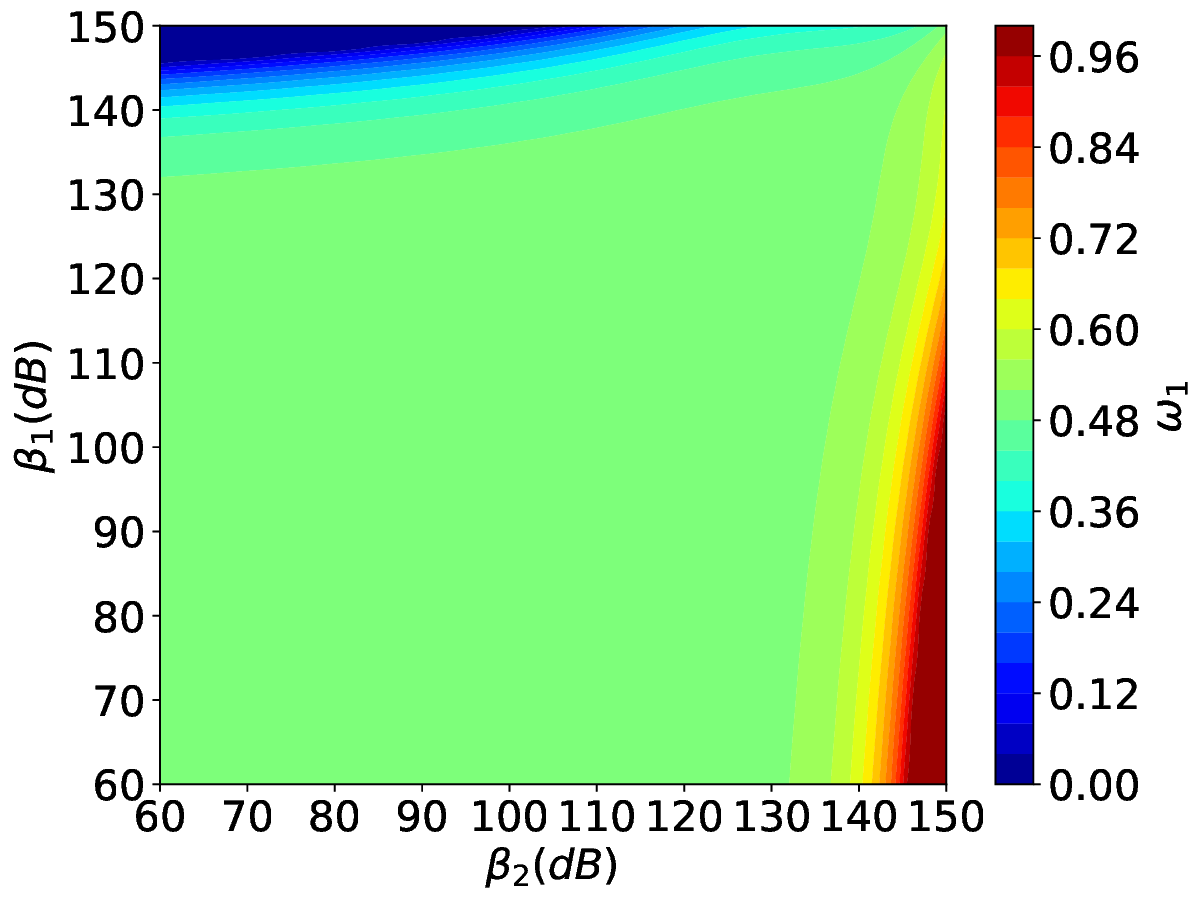}
\caption{Fraction of power allocated by DAPA-FPDA to UE 1 ($\omega_1$) for a 2 UE case.}
\label{fig:omega_distribution_DAPA_FPDA}
\end{figure} 

\begin{figure}[!t]
\centering
\includegraphics[width=0.8\columnwidth]{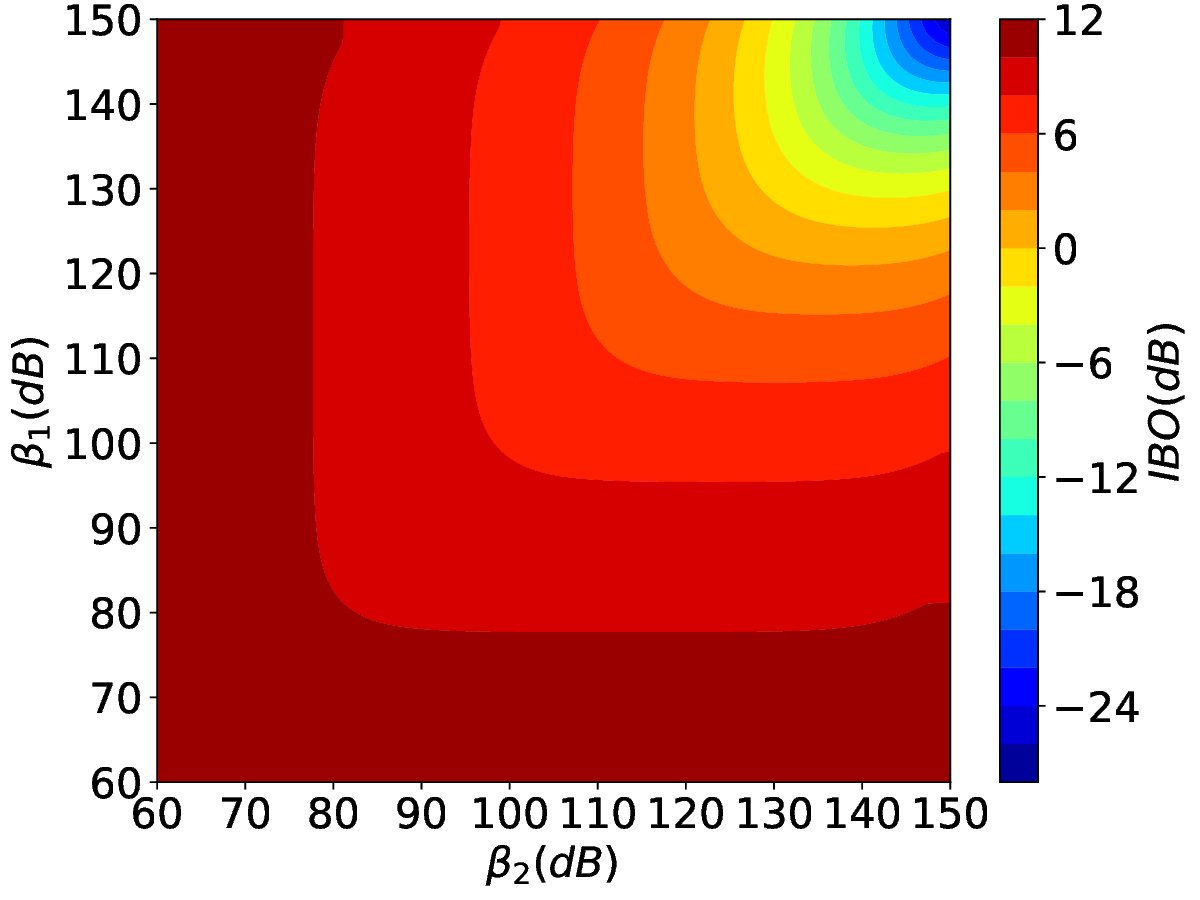}
\caption{Optimal IBO for DAPA-FPDA for a 2 UE case.}
\label{fig:IBOopt}
\end{figure}  

More interesting is the off-diagonal performance of DAPA-FPDA algorithm, i.e., when there is a difference in path loss between UEs. It can be observed from Fig. \ref{fig:ratio_DAPA_FPDA_REFE} that the highest sum-rate gain with respect to the REF-E algorithm is obtained in the extreme case, when one UE is close to the BS (path loss around 60 dB) and the other UE far away (path loss around 150 dB). This allows to increase the sum-rate by approx. 80\%, which is obtained by allocating most of the power to the UE with better channel (lower path loss) (see Fig. \ref{fig:omega_distribution_DAPA_FPDA}).  Most interestingly, in quite wide range of path losses, i.e., if both UEs have path loss below approx. 130 dB, the equal allocation ($\omega_1=\omega_2=0.5$) is optimal. This stems from the typical gains of water filling: achieved under low SNR regime, if there are significant differences between multiple considered channels.

One might ask how the IBO is selected if there are multiple UEs of different wireless channel path loss. In Fig. \ref{fig:IBOopt} nearly L-shaped strips of equal IBO values are visible. Mostly the smaller path loss, ergo: the UE closer to the BS, \emph{decides} on the IBO. This can be attributed to the sum-rate maximization, where it is reasonable to maximize SNDR of the closest UE because it provides the highest rate out of all allocated UEs. Most interestingly, even though the IBO changes rapidly with the best UE path loss, the power distribution among UEs is very stable and equal for all UEs with path loss lower than 130 dB,  which is observable in Fig. \ref{fig:omega_distribution_DAPA_FPDA}.

\subsection{Multi-UE Heterogeneous Path Loss}
\label{subsec:multi_ue_het_loss}

In this scenario, the sum-rate is evaluated for $P_{\mathrm{max}} = 100$~mW, where UEs are uniformly distributed within a $2$ km coverage radius of a circular cell. For each $M = [64, 512]$ antennas, we generate $1000$ random locations for $K=60$ UEs.  

\begin{figure}[!t]     
\centering
\includegraphics[width=0.95\columnwidth]{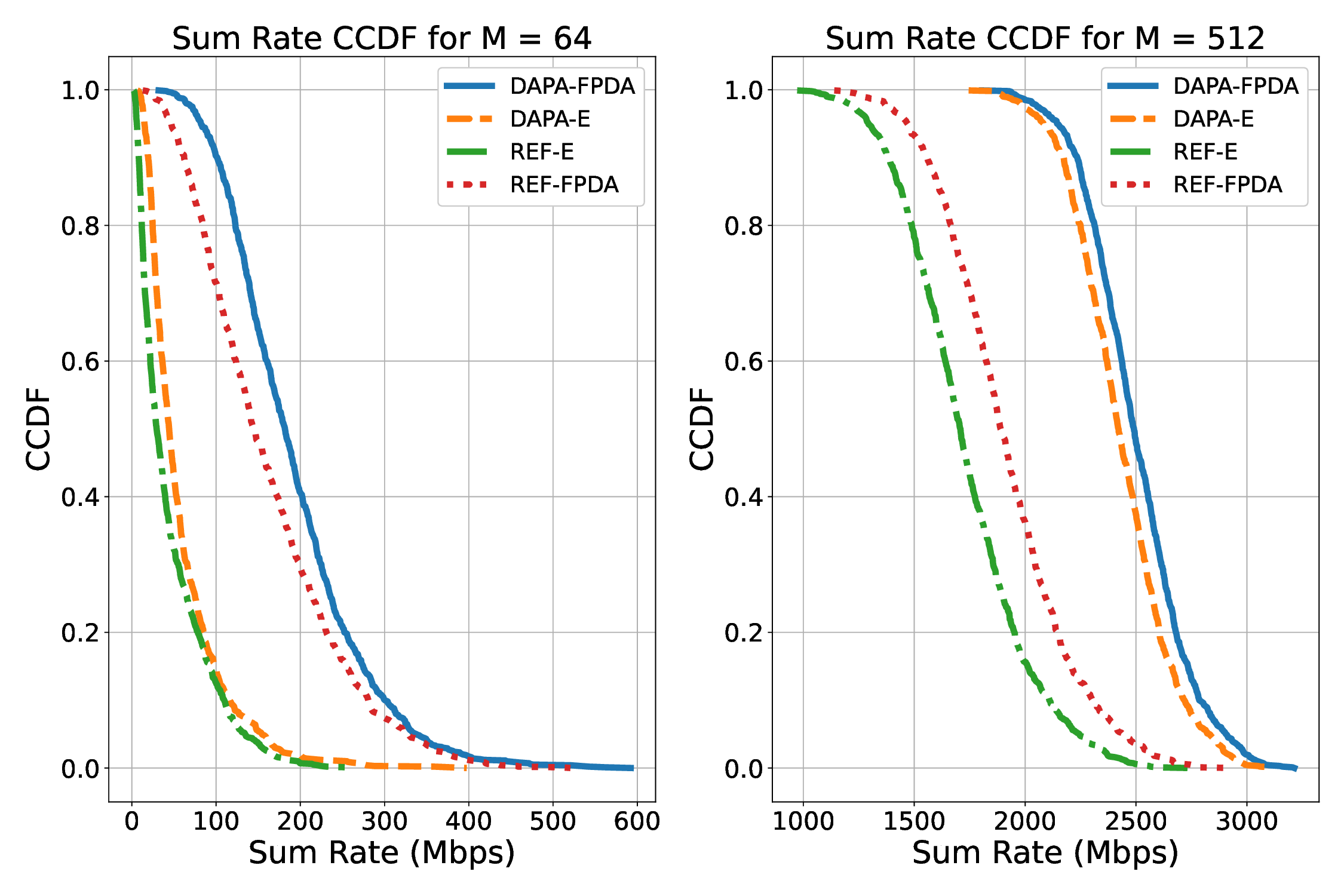} 
    \caption{CCDF of sum-rate for $K = 60$ UEs of random locations in a 2 km radius cell.}
    \label{fig:CCDF_rate}
\end{figure}
Figure \ref{fig:CCDF_rate} shows the Complementary Cumulative Distribution Function (CCDF) of the sum-rate. First, it can be observed that a higher sum-rate is achieved as the number of antennas increases. Secondly, the proposed DAPA-FPDA algorithm outperforms the DAPA-E, REF-E, and REF-D algorithms irrespective of the number of antennas used. The highest gains with respect to the REF-E solution are observed for $M=512$ with a median sum-rate gain of about 50 \%. The gain with respect to the REF-FPDA solution is slightly lower, however, still significant, i.e. approx. 40 \% for the median sum-rate. 
It is interesting to observe that while DAPA-E is nearly optimal for $M=512$ antennas providing only a slight sum-rate loss with respect to DAPA-FPDA algorithms, it is outperformed by the REF-FPDA algorithm for the $M=64$ case. This shows that depending on the use case, both unequal power allocation (through FPDA) and optimal mean power allocation (DAPA) play a dominant role in maximizing the sum-rate. 

\begin{figure}[!t]        
\centering
\includegraphics[width=0.95\columnwidth]{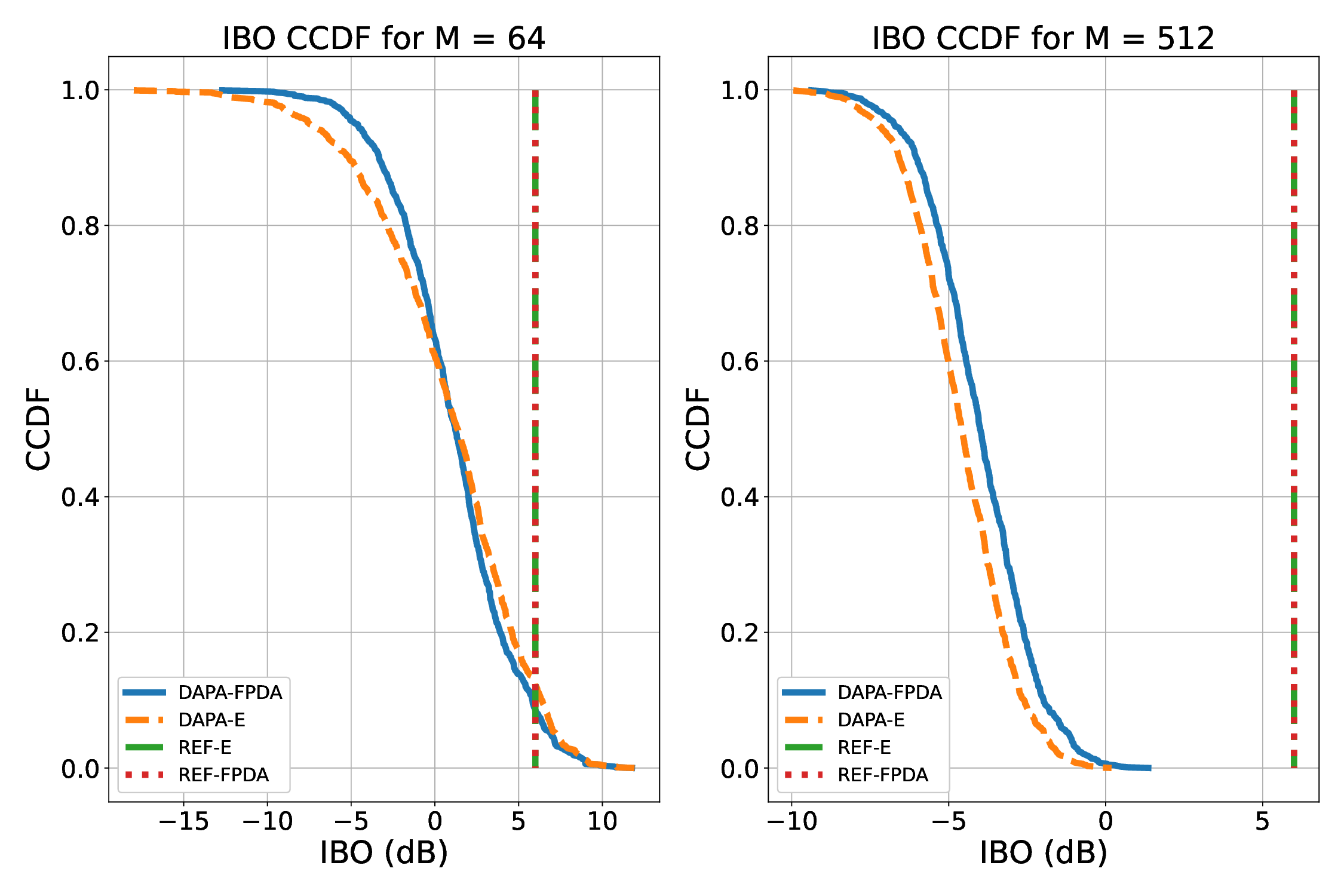} 
    \caption{CCDF of IBO for $K = 60$ UEs of random locations in a 2 km radius cell.}
    \label{fig:CCDF_IBO}
\end{figure}
The mechanism of sum-rate gain is partially revealed by Fig.~\ref{fig:CCDF_IBO} showing the CCDF of the IBO value. It is visible that almost always DAPA-FPDA and DAPA-E use lower IBO, allowing for stronger signal clipping than the reference algorithms. Moreover, it is observable that the optimal IBO distribution in both DAPA-FPDA and DAPA-E is different, i.e., a single iteration of alternating optimization is not enough as the IBO is further modified in the next iterations. In most cases, the M-MIMO configuration with $M=512$ antennas uses lower IBO in DAPA-FPDA than for $M=64$ antennas. This contradicts the IBO selection conclusions for homogeneous UE scenario, as shown in Fig. \ref{fig:rate_ibo}.
This is caused by the path loss $\beta_k$ varying among UEs. The variation also results in varying power allocated to each UE, as visible in Fig. \ref{fig:ccdf_max_omega}. It shows the CCDF of the maximal $\omega_k$ value. It can be observed that each $\omega_k$ equals $1/60\approx 0.017$ when equal power is allocated. However, it is visible especially for $M=64$ antennas that a single UE takes at least 30\% of total allocated power in most cases. For a higher number of antennas ($M=512$), the power distribution among the UEs is rather equal.
This can be justified from Fig. \ref{fig:omega_distribution_DAPA_FPDA}, i.e., there is a need to differentiate power allocation among UEs only when some of them have very poor propagation conditions. By increasing the number of antennas a given UE can expect to achieve a higher SNR value, reducing the need for varying power among allocated UEs.
Moreover, the power allocation is more polarized in REF-FPDA than in the DAPA-FPDA case.
\begin{figure}[!t]
    \centering     
    \includegraphics[width=0.95\columnwidth]{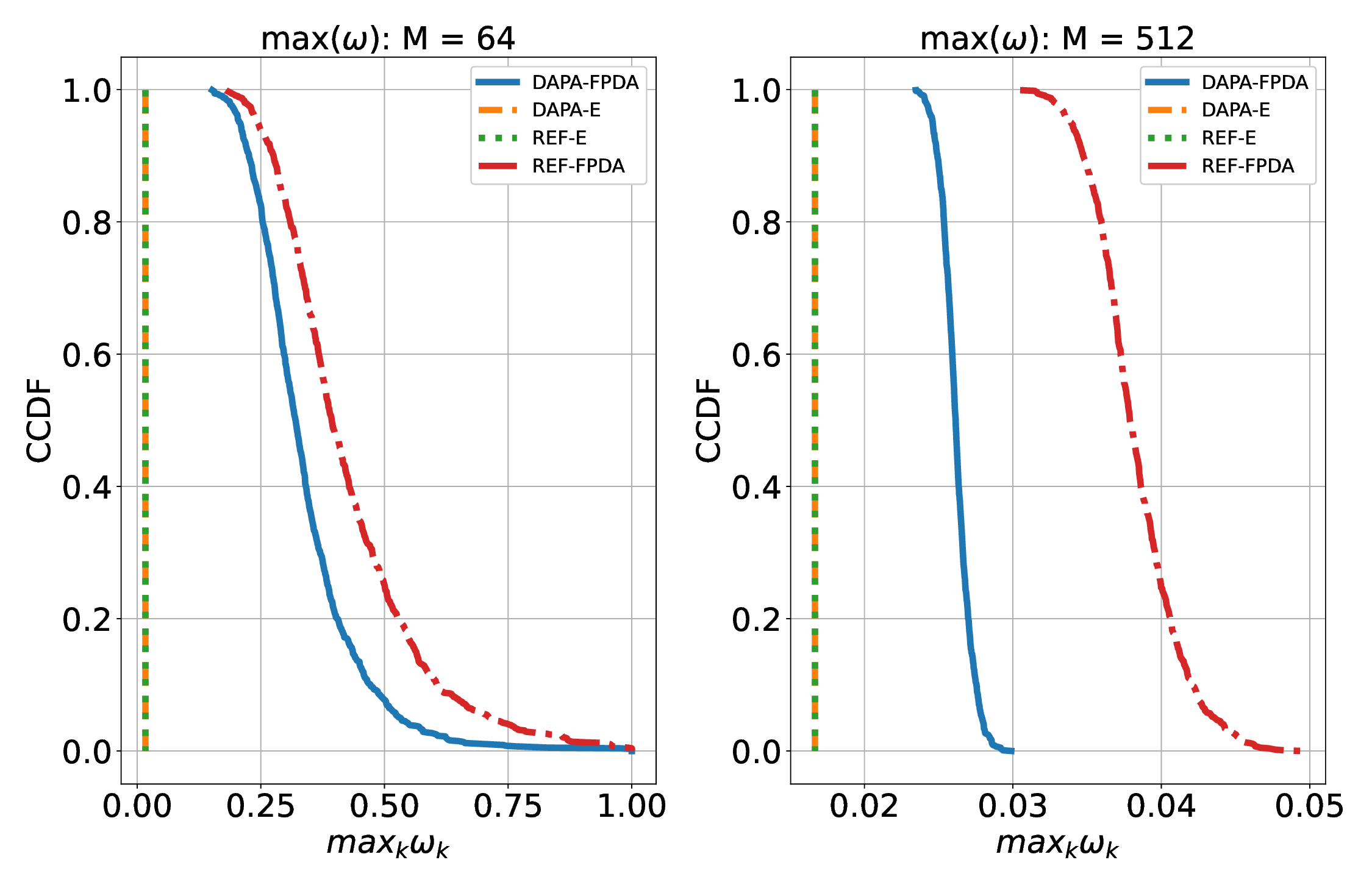}    
    \caption{CCDF of $\max_k \omega_k$ for $K=60$ UEs of random locations in a 2 km radius cell.}
    \label{fig:ccdf_max_omega}
\end{figure}

\begingroup
\color{blue}
\subsection{Impact of non-predistorted PA}
\label{subsec:rapp_compare}
While the soft limiter PA, used both in the system model and optimization, is an ultimate characteristic of a PA, possibly obtained if DPD is performed \cite{raich2005_optimal_nonlinearity}, it is possible that the M-MIMO  system will operate under less idealistic PA without DPD. This can be modeled by the Rapp model with smoothness parameter $p=2$ \cite{ochiai2013analysis}. In this subsection, while the resource allocation assumes a soft limiter (or Rapp model of $p\to \infty$), the achievable sum-rate is evaluated using Rapp model with $p=2$. This requires numerical calculation of $\lambda$ and $D$ using integrals defined in \eqref{eq_distortion_SISO_Rapp} and \eqref{eq_lambda_Rapp}.

\begin{figure}[!t]     
\centering
\includegraphics[width=0.95\columnwidth]{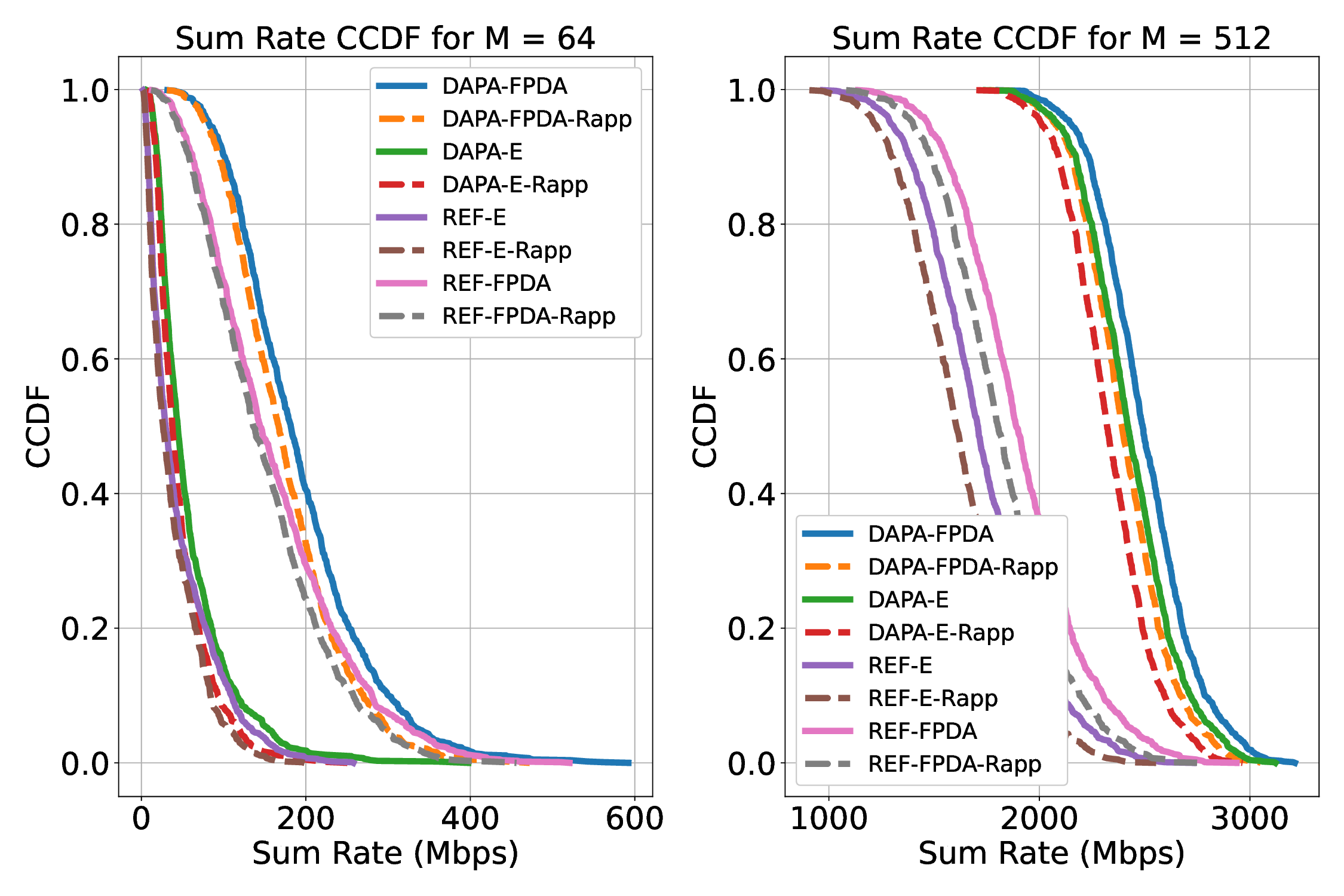} 
    \caption{Comparison between CCDFs of sum-rate for sof-limiter and Rapp model with $p=2$ for $K = 60$ UEs of random locations in a 2 km radius cell.}
    \label{fig:CCDF_rate_Rapp}
\end{figure}
Fig. \ref{fig:CCDF_rate_Rapp} shows the sum-rate CCDF for $M=64$ and $M=512$ antennas, where the same scenario as described in Sec. \ref{subsec:multi_ue_het_loss} is simulated. For each algorithm provided in Table \ref{tab: algos}, two PA cases have been considered: \emph{A)} when PA is modeled via a soft limiter\textemdash giving results identical to the ones from Sec. \ref{subsec:multi_ue_het_loss}, marked with solid lines, and \emph{B)} when PA is modeled via Rapp model\footnote{\textcolor{blue}{Note that all evaluated algorithms assume soft-limiter for the power allocation as Rapp model is not analytically tractable, which is discussed in Sec. \ref{subsec: PA-NL}.}}\textemdash marked with dashed lines. It can be observed that in all the cases the achieved sum-rate is degraded by a small percent. While this can be result of the optimization considering soft-limiter PA model, it is also caused by lower SDR values achieved when PA characteristics other than soft-limiter\cite{raich2005_optimal_nonlinearity} are employed. Most importantly, not only is the observed sum-rate degradation relatively low, but the order of algorithms also does not change, i.e. DAPA-FPDA still outperforms all other analyzed allocation schemes. Therefore, even under PA characteristics more severe than the soft-limiter in terms of nonlinear distortion, it is beneficial to perform the allocation utilizing DAPA-FPDA.

\subsection{Impact of Imperfect Channel State Information (CSI)}
\label{subsec:impCSI}

We quantify the impact of imperfect CSI on the performance of the proposed distortion-aware power allocation framework by comparing the sum-rate achieved under perfect and imperfect CSI assumptions, where an additive error model is adopted. Assuming that $\boldsymbol{h}_{k,n} \in \mathbb{C}^{M\times1}$ represent the true DL channel vector on sub-carrier $n$ for UE $k$, whereas the BS employs the imperfect channel estimate  $\boldsymbol{\hat{h}_{k,n}}$, the relation between the true and estimated channels can be represented as \cite{massivemimobook, 8641436}
\begin{equation}
    \boldsymbol{h} _{k,n} = \boldsymbol{\hat{h}}_{k,n} - \boldsymbol{e}_{k, n}, 
\end{equation}
where, $\boldsymbol{e}_{k, n} \sim \mathcal{CN}(0,\,\delta_k \beta_k\boldsymbol{I}_M)$ is the channel estimation error, $\delta_k \in (0, 1)$ is the normalized error factor  
and $\boldsymbol{\hat{h}}_{k,n} \sim \mathcal{CN}(0, (1-\delta_k)\beta_k\boldsymbol{I}_M)$ is the MMSE based estimated channel. Adopting this $\delta_k$-based model, the imperfect CSI SINDR for UE k becomes 
\begin{equation}
    \gamma_k^{\mathrm{ZF_{iCSI}}} =  \frac{(M - K)\cdot \lambda \cdot p_{k}\cdot \beta_{k} (1 - \delta_k)}{\sigma_k^2 + \beta_{k} D +  \lambda \cdot \beta_{k} \cdot \delta
    _k\cdot \sum_{k' \neq k}p_{k'}},
\label{eq: SINR5}
\end{equation}
where the desired signal power scales with $(1 - \delta_k)$ and the residual inter-UE interference scale with $\beta_{k}\delta_k$. The inter-UE interference appears as a result of imperfect zero forcing. 

In the following, to evaluate the robustness of the proposed alternating DAPA-FPDA framework, we adopt the same scenario as described in Sec. \ref{subsec:multi_ue_het_loss}.
For a 2\,km 5G macrocell with 60 UEs and the large-scale fading model 
$L = 36.7 \log_{10}(d) + 22.7 + 26 \log_{10}(f_c)$ at 
$f_c = 3.5\,\mathrm{GHz}$, the channel estimation error with error factor \cite{8641436} 
$\delta_k = \frac{1}{1 + N_p \rho_{\mathrm{ul}}\beta_k}$\textemdash with $N_p$ as the pilot sequence length and $\rho_{ul}$ as uplink SNR\textemdash was computed using $\beta_k = 10^{-L/10}$, $N_p = 60$, $\rho_{ul} = 23$ dBm and bandwidth of $20$ MHz \cite{8641436}. The resulting $\delta_k$ ranges 
from $10^{-3}$ ($\approx -30\,\mathrm{dB}$) for cell-center UEs up to 
$\sim 0.2$ ($\approx -7\,\mathrm{dB}$) at the 2\,km cell edge, with a 
spatial average of $\sim 0.05$--$0.1$ across uniformly distributed UE. 

\begin{figure}[!t]     
\centering
\includegraphics[width=0.95\columnwidth]{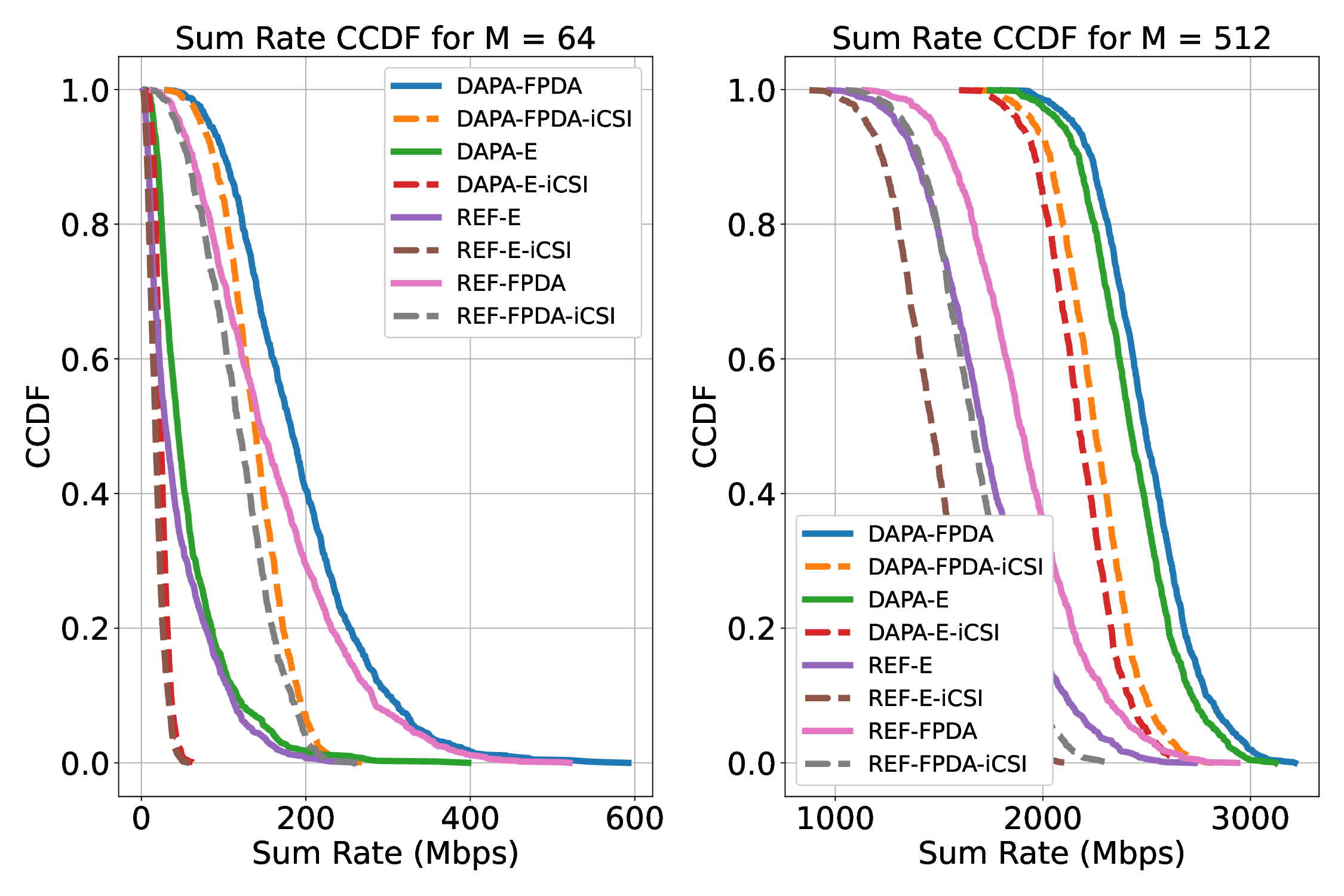} 
    \caption{CCDF of sum-rate for $K = 60$ UEs of random locations in a 2 km radius cell with $\delta_k = \delta = 0.1$ (same for each UE). 
    }
    \label{fig:CCDF_rate_model_1}
\end{figure}
Observe that since most UEs fall in the $0.03$--$0.15$ range and mid-to-cell-edge performance drives system-level behavior, assuming a fixed 
$\delta = 0.1$, corresponding to relatively accurate channel estimates (i.e., 90\% of the channel power is captured in the estimate), for all UEs provides a conservative yet representative approximation of the uplink channel estimation error in this scenario.
With this fixed estimation error factor $\delta = 0.1$ for all UEs, the results, as demonstrated in Fig. \ref{fig:CCDF_rate_model_1}, show that the sum-rate under imperfect CSI follows the same qualitative trends as in the perfect CSI case, however, with a noticeable degradation in magnitude. The performance loss is primarily due to the reduced effective signal gain and the re-emergence of inter-UE interference resulting from the breakdown of perfect ZF under estimation error. Nevertheless, the observed gap confirms that the proposed algorithm retains a substantial portion of its performance even when CSI is imperfect, demonstrating a level of robustness despite being optimized under ideal conditions. These findings highlight the importance of developing CSI-aware extensions, while also validating the utility of the current method as a practical baseline.

\endgroup

\section{Conclusion and Future Work}
\label{sec:conc}

In this work, we have analyzed the impact of non-linear PA characteristics on the sum-rate of a M-MIMO OFDM system employing a wideband signal model and utilizing ZF precoder in the downlink, where the PA is modeled by a soft-limiter model. The proposed power allocation optimization problem does not require the total transmit power to be constrained as the distortion introduced by the PA self-constraints the result. To solve the non-convex optimization problem an alternating optimization framework is used. Our simulation results demonstrate the superior performance of the proposed DAPA-FPDA algorithm relative to existing distortion-neglecting solutions. Most interestingly, in many cases it is optimal to significantly clip the transmitted signal (using negative IBO values) even though the non-linear distortion is treated as white noise. 

\textcolor{blue}{For future work, firstly, an advanced distortion-aware receiver, e.g., \cite{Wachowiak2023}, should be considered, which is able to partially remove the distortion or use nonlinear distortion positively (as a source of frequency diversity). This should result in an even stronger clipping being optimal. Secondly, while this work relies on channel hardening to abstract subcarrier effects under i.i.d. Rayleigh fading, it would be beneficial to generalize the optimization problem to spatially correlated channels. However, a redefinition of the SNDR metric will be required, potentially using the distortion estimation framework from \cite{mollen_spatial_char}. Additionally, scenarios where subcarrier scheduling may improve performance should be considered, which requires incorporating per-subcarrier UE assignment, redefining the SNDR and extending the problem with discrete variables and per-subcarrier distortion modeling.
Thirdly, incorporating minimum rate (QoS) constraints to ensure fairness and service guarantees is also relevant. However, unlike linear SINR systems where Perron–Frobenius theory applies to ensure feasibility, the feasibility analysis in our work becomes intractable due to the non-linear distortion-coupled SINDR.} Finally, the other options are MRT consideration, extension to multi-cell scenario, or energy efficiency maximization. 

\appendices

\section{Proof of Proposition \ref{prop1}}
\label{sec: non-convexity}
This section shows that the sum-rate in \eqref{eq:optprob} is non-convex. It is known from convex optimization theory that if the eigenvalues of the Hessian matrix are non-negative, then the Hessian is positive-semi-definite and the function is convex. 
Whereas, if the eigenvalues are non-positive, then the Hessian is negative-semi-definite and the function is concave  
\cite{boyd2004convex}. However, if the eigenvalues are both positive and negative, then the function is non-convex. Therefore, to show the non-convexity of the sum-rate function, let us consider a two-UE scenario, where the allocated power for UE 1 and 2 are $p_1$ and $p_2$, respectively. Then, the sum-rate can be defined as
\begin{equation}
    \begin{aligned}
      f(\boldsymbol{p}) = B \left( \log_2\left(1 + \frac{(M-2)\lambda p_1 \beta_1}{\sigma^2_1 + \beta_1 D}\right) \right. + \\
       \left. \log_2\left(1 + \frac{(M-2)\lambda p_2 \beta_2}{\sigma^2_2 + \beta_2 D}\right) \right), 
    \end{aligned}
    \label{eq:sum_rate_f}
\end{equation}
where $\lambda$ and $D$ are functions of $\Psi = \frac{M P_{\mathrm{max}}}{p_1 + p_2}$ defined in \eqref{eq_lambda} and \eqref{eq_distortion_SISO}, respectively. The Hessian of $f$ can be computed as
\begin{equation}
H(f) = \begin{bmatrix}
    \frac{\partial^2 f}{\partial p_1^2} & \frac{\partial^2 f}{\partial p_1 \partial p_2} \\
    \frac{\partial^2 f}{\partial p_2 \partial p_1} & \frac{\partial^2 f}{\partial p_2^2} 
\end{bmatrix},
\label{eq:Hessian}
\end{equation}
where the partial derivatives of $f$, i.e. $\frac{\partial f}{\partial p_1}$, $\frac{\partial^2 f}{\partial p_1^2}$ and $\frac{\partial^2 f}{\partial p_1 p_2}$ , can be computed by \eqref{eq:pd_p1}, \eqref{eq:pd_p1_p_1}, and \eqref{eq:pd_p1_p2}, respectively and the partial derivatives $\frac{\partial^2 f}{\partial p_2 p_1}$ and $\frac{\partial^2 f}{\partial p_2^2}$ can be derived similarly.
\begin{equation}
    \begin{aligned}
      \frac{\partial f}{\partial p_1} = B(M-2) \left( \frac{U}{V} + \frac{W}{X}\right), 
    \end{aligned}
    \label{eq:pd_p1}
\end{equation}
where, 

$U = \lambda - p_1 \frac{\partial \lambda}{\partial p_1} - \frac{\lambda p_1 \frac{\partial D}{\partial p_1}}{\frac{\sigma_1^2}{\beta_1} + D}$, $V = \frac{\sigma_1^2}{\beta_1} + D + (M-2)\lambda p_1$, 

$W = p_2\left(\frac{\partial \lambda}{\partial p_1} - \frac{\lambda \frac{\partial D}{\partial p_1}}{\frac{\sigma_2^2}{\beta_2}+D}\right)$, and 
$X = \frac{\sigma_2^2}{\beta_2}+D+(M-2)\lambda p_2$.

\begin{equation}
    \begin{aligned}
      \frac{\partial^2 f}{\partial p_1^2} = B(M-2) \left(\frac{S}{T} - \frac{U}{T^2} +  \frac{W}{X} - \frac{Y}{X^2}
      \right),
    \end{aligned}
    \label{eq:pd_p1_p_1}
\end{equation}
where, 

$S = -p_1 \frac{\partial^2 \lambda}{\partial p_1^2} - \frac{\left(\lambda \frac{\partial D}{\partial p_1} + p_1 \frac{\partial \lambda}{\partial p_1} \frac{\partial D}{\partial p_1} + \lambda p_1 \frac{\partial^2 D}{\partial p_1^2} - \frac{\lambda p_1 \frac{\partial D}{\partial p_1}}{\frac{\sigma_1^2}{\beta_1} + D}\right)}{\frac{\sigma_1^2}{\beta_1} + D}$, 

$T = \frac{\sigma_1^2}{\beta_1} + D + (M-2)\lambda p_1$, 

$U = \left(\lambda -  \frac{p_1 \partial \lambda}{\partial p_1} - \frac{\lambda p_1 \frac{\partial D}{\partial p_1}}{\frac{\sigma_1^2}{\beta_1}+D}\right) \left(\frac{\partial D}{\partial p_1} + (M-2) (\lambda +  \frac{p_1 \partial \lambda}{\partial p_1} )\right)$, 

$W = p_2 \left(\frac{\partial^2 \lambda}{\partial p_1^2} - \frac{\frac{\partial \lambda}{\partial p_1} \frac{\partial D}{\partial p_1} + \lambda \frac{\partial^2 D}{\partial p_1^2} + \frac{\lambda (\frac{\partial D}{\partial p_1})^2}{\frac{\sigma_2^2}{\beta_2}+D}}{\sigma_2^2 + D}\right)$, 

$X = \frac{\sigma_2^2}{\beta_2} + D + (M-K)\lambda p_2$, 

$Y = \left(\frac{\partial \lambda}{\partial p_1} - \frac{\lambda \frac{\partial D}{\partial p_1}}{\frac{\sigma_2^2}{\beta_2}+D}\right)\left(\frac{\partial D}{\partial p_1} + (M-2) \frac{\partial \lambda}{\partial p_1} p_2\right)$.

\begin{equation}
    \begin{aligned}
      \frac{\partial^2 f}{\partial p_1 p_2} = B(M-2) \left(\frac{S}{T} - \frac{U}{T^2} +  \frac{W}{X} - \frac{Y}{X^2}
      \right),
    \end{aligned}
    \label{eq:pd_p1_p2}
\end{equation}
where, 

$S = \frac{\partial \lambda}{\partial p_2} - p_1 \frac{\partial \lambda}{\partial p_1 \partial p_2} - \frac{\left(p_1 \frac{\partial \lambda}{\partial p_2} \frac{\partial D}{\partial p_1}+ p_1 \lambda \frac{\partial D}{\partial p_1 \partial p_2} - \frac{p_1 \lambda \frac{\partial D}{\partial p_1} \frac{\partial D}{\partial p_2}}{\frac{\sigma_1^2}{\beta_1}+D}\right)}{\frac{\sigma_1^2}{\beta_1}+D}$, 

$T = \frac{\sigma_1^2}{\beta_1} + D + (M-2)\lambda p_1$, 

$U = \left(\lambda - p_1 \frac{\partial \lambda}{\partial p_1} - \frac{\lambda p_1 \frac{\partial D}{\partial p_1}}{\frac{\sigma_1^2}{\beta_1}+D}\right) \left(\frac{\partial D}{\partial p_2} + (M-2) \frac{\partial \lambda}{\partial p_2}p_1\right)$,

$W = \left(\frac{\partial \lambda}{\partial p_1} - \frac{\lambda \frac{\partial D}{\partial p_1}}{\frac{\sigma^2_2}{\beta_2}+D}\right) +  \\ 
\phantom{\qquad\qquad\qquad} p_2 \left(\frac{\partial \lambda}{\partial p_1 \partial p_2} - \frac{\frac{\partial \lambda}{\partial p_2} \frac{\partial D}{\partial p_1} + \lambda \frac{\partial D}{\partial p_1 \partial p_2} + \frac{\lambda \frac{\partial D}{\partial p_1} \frac{\partial D}{\partial p_2}}{\frac{\sigma_2^2}{\beta_2}+D}}{\sigma_2^2 + D}\right)$, 

$X = \frac{\sigma_2^2}{\beta_2} + D + (M-K)\lambda p_2$, 

$Y = p_2\left(\frac{\partial \lambda}{\partial p_1} - \frac{\lambda \frac{\partial D}{\partial p_1}}{\frac{\sigma_2^2}{\beta_2}+D}\right)\left(\frac{\partial D}{\partial p_2} + (M-2)(\lambda + p_2 \frac{\partial \lambda}{\partial p_2})\right)$.

\medskip
To show that \eqref{eq:sum_rate_f} is non-convex, it is sufficient to find one such scenario resulting in a positive and a negative eigenvalue of the Hessian matrix in \eqref{eq:Hessian}, where the eigenvalues of $H$ can be found by solving $|H - \xi I| = 0$, where $|\;|$ represents the determinant, $I$ is the identity matrix of the same order as the Hessian $H$ and $\xi$ represents the eigenvalues. We compute the Hessian numerically and show that for $P_{\mathrm{max}} = 10$ mW, $\eta = \frac{2}{3}$, $M =  64$, $B = 18$ MHz, $\sigma_k^2 = \num{5.97e-14}$, $\beta_1 = \num{1e-11}$, and $\beta_2 = \num{1e-7}$, the eigenvalues are $\xi_1 = \num{-1.44e1}$ and $\xi_2 = \num{8.33e4}$ and conclude that the sum-rate in \eqref{eq:optprob} is non-convex.

\section{Proof of Proposition \ref{prop:equivalence} }
\label{sec:equivalence}
In this section we provide a detailed proof showing that the solutions to \eqref{eq:TPA} and \eqref{eq:optprob5} find a stationary point of \eqref{eq:optprob3}, which corresponds to finding the stationary point of \eqref{eq:optprob}. 

First, we provide the KKT conditions for the problem in \eqref{eq:optprob3}. Let us define the goal function of \eqref{eq:optprob3} as 
\begin{equation}
   f(\boldsymbol{p}) =  \sum_k B \log_2\left(1 + \frac{(M - K) \cdot \lambda \cdot \omega_{k} P \cdot \beta_{k}}{\sigma_k^2 + \beta_{k} D} \right),
\end{equation}
where $\boldsymbol{p} =[\omega_1P,\dots,\omega_KP]$. The Lagrangian function for the problem in \eqref{eq:optprob3} can be written as 
\begin{alignat}{2}
        L(\boldsymbol{\omega}, P, \alpha,  \boldsymbol{\mu}, \nu) =& f(\omega_1 P, \omega_2 P, \ldots, \omega_K P) + \alpha (-P)  \nonumber \\
        & + \sum_{k=1}^K \mu_k(-\omega_k) + \nu (1 -\sum_k \omega_k).
        \label{eq_Lagrangian_both_variables}
\end{alignat}
Assuming $\boldsymbol{\mu}^* = [\mu_1^*,\dots, \mu_K^*]$, $\alpha^*$, $\nu^*$ and $\boldsymbol{p}^* = [\omega_1^*P^*,\dots, \omega_K^*P^*]$ are the optimal solutions for the dual and primal problem, the primal feasibility, dual feasibility and complementary slackness conditions, respectively, for all $k$ are:
    \begin{align}
    &\text{C1}: w^*_k \geq 0, P^* \geq 0, \text{C2}: \mu^*_k \geq 0, \alpha^* \geq 0,
    \nonumber
    \\
    &\text{C3}: \mu^*_k \omega^*_k = 0, \alpha^* P^* = 0,  \nu^* (\sum_k \omega^*_k - 1) = 0. 
    \label{eq_KKT_full}
    \end{align}
The last KKT condition, stationarity, states that $\frac{\partial L}{\partial \omega^*_k}= 0$ and $\frac{\partial L}{\partial P^*}= 0$.
 Then, by using chain rule of derivative 
\begin{equation}
    \frac{\partial  L}{\partial \omega_k^{*}}=\frac{\partial f}{\partial p_k^{*}} \frac{\partial p_k^{*}}{\partial \omega_k^{*}}-\mu_k^{*} - \nu^*=\frac{\partial f}{\partial p_k^{*}}P^* -\mu_k^{*} - \nu^*=0.
    \label{eq:stationarity1}
\end{equation}
Following the same approach, the stationary condition with respect to $P^*$ is
\begin{equation}
    \frac{\partial L}{\partial P^{*}} = \sum_{k=1}^K  \frac{\partial f}{\partial p_k^{*}} \omega^*_k - \alpha^* = 0.
     \label{eq:stationarity22}
\end{equation}

Next, we provide the KKT conditions separately for the DAPA and FPDA sub-problems, which are solved through an alternating optimization framework. 

First, observe that the Lagrangian function for the DAPA and FPDA sub-problems in \eqref{eq:TPA} and \eqref{eq:optprob5} can be written as 
\begin{equation}
    L_p(P,\alpha) = f(\omega_1 P, \omega_2 P, \ldots, \omega_K P) -  {\alpha} P 
\end{equation}
and  
\begin{alignat}{2}
    L_{\boldsymbol{\omega}}(\boldsymbol{\omega}, \nu,\boldsymbol{\mu}) =& f(\omega_1 P, \omega P, \ldots, \omega_K P)  \nonumber \\ & +{\nu} (1 - \sum_{k=1}^K \omega_k) - \sum_{k=1}^K {\mu_k} \omega_k, 
\end{alignat}
respectively. Then, let us assume that $\tilde{\alpha} $ and $\tilde{P}$ are the optimal solutions for the dual and primal problem DAPA. Similarly, $\tilde{\omega}_k$, $\tilde{\nu}$, $\tilde{\mu_k}$ are the optimal solutions for the FPDA sub-problem. Based on these, the optimal k-th UE power can be defined as $\tilde{p}_k=\tilde{\omega}_k \tilde{P}$. 
Next, for the DAPA sub-problem, the KKT conditions can be stated as
\begin{subequations}	
    \begin{align*}
    &\text{E1}: \tilde{P} \geq 0, \text{E2}: \tilde{\alpha} \geq 0, \text{E3}: \tilde{\alpha} \tilde{P} = 0,
    \end{align*}
\end{subequations}
with the stationarity condition being
\begin{equation}
     \frac{\partial L_p}{\partial \tilde{P}} = \sum_{k=1}^K  \frac{\partial f}{\partial \tilde{p}_k }\tilde{\omega}_k
    - \tilde{\alpha} = 0.
    \label{eq:stationarity3}
\end{equation}
Similarly, for the FPDA sub-problem, the KKT conditions can be stated as
\begin{alignat}{2}
    &\text{F1}: \tilde{\omega}_k \geq 0 \; \forall k, \; \text{F2}: \sum_{k=1}^{K} \tilde{\omega}_k - 1 =0, \text{F3}: \tilde{\mu}_k \geq 0 \; \forall k; \nonumber \\ 
    &\text{F4}: \tilde{\mu}_k \omega_k = 0 \; \forall k;  \tilde{\nu} (1 -\sum_k \tilde{\omega}_k) = 0 \nonumber
\end{alignat}
and the stationary condition F5 yields
\begin{equation}
    \frac{\partial  L_{\boldsymbol{\omega}}}{\partial \tilde{\omega}_k}=\frac{\partial f}{\partial \tilde{p}_k}\tilde{P} - \tilde{\nu} - \tilde{\mu}_k = 0.
    \label{eq:stationarity4}
\end{equation}
Note that we assume that both DAPA and FDPA sub-problems are at their solution, which is obtained by alternating optimization, so that in DAPA and FPDA the KKT conditions $\tilde{\omega}_k$ and $\tilde{P}$ are used, respectively. 

Let us state the hypothesis that primal and dual solutions of FPDA and DAPA sub-problems meet all the KKT conditions of \eqref{eq:optprob3} defined in (\ref{eq_KKT_full}). In other words, $\tilde{\omega}_k$, $\tilde{P}$, $\tilde{\mu}_k$, $\tilde{\alpha}$, $\tilde{\nu}$ meet (\ref{eq_KKT_full}) if substituted for ${\omega}^{*}_k$, ${P}^{*}$, ${\mu}^{*}_k$, ${\alpha}^{*}$, ${\nu}^{*}$, respectively. 

Based on F1 and E1 it is visible that C1 is met. Similarly, C2 can be proved by using F3 and E2. The complementary slackness  C3 is obtained directly from E3 and F4. As such only stationarity condition is left. As the derivative in (\ref{eq:stationarity4}) equals 0, based on the same structure of (\ref{eq:stationarity1}) its value will be zero as well. Similarly, as it is assumed that derivative in (\ref{eq:stationarity3}) equals 0, the same is achieved for (\ref{eq:stationarity22}). 

This proves the hypothesis that if the primal and dual solutions are obtained in DAPA and FDPA via alternating optimization, the solution also meets KKT conditions of \eqref{eq:optprob3}.

\section{Proof of Lemma \ref{lem:DAPA_1}}
\label{sec:proof_lemma_1}

The partial derivative $\frac{\partial \tilde{R}_k}{\partial P}$ can be calculated as
\begin{align}
\label{eq_der_Rk}
    \frac{\partial \tilde{R}_k}{\partial P}=&
    \frac{B}{\log(2) \left( 
    1+\frac{(M-K)\lambda\omega_k P \beta_k}{\sigma_k^2+\beta_k D}
    \right)
    }
    \frac{(M-K)\omega_k\beta_k}{\left(\sigma_k^2+\beta_k D\right)^2 }
    \nonumber
    \\&
    \cdot\left(
    \left(
    \frac{\partial \lambda}{\partial P}P
    +\lambda
    \right)
    \left( \sigma_k^2+\beta_k D \right)
    -\lambda P \beta_k \frac{\partial D}{\partial P}
    \right),
\end{align}
where 
\begin{align}
    \frac{\partial \lambda}{\partial P}=&
    -\frac{\Psi\sqrt{\lambda}}{P}
    \left( 
    e^{-\Psi}+\frac{1}{2}\sqrt{\frac{\pi}{\Psi}} \textrm{erfc} \left( \sqrt{\Psi} \right)
    \right),
\end{align}
and 
\begin{align}
     \frac{\partial D}{\partial P}=&
     \eta \left(1-  e^{-\Psi}-\lambda 
     -P \frac{\partial \lambda}{\partial P}-\Psi e^{-\Psi} \right),
\end{align}
recalling that IBO is redefined as $\Psi=MP_{max}/ P $.
After substitution and some operations on \eqref{eq_der_Rk}, $\frac{\partial \tilde{R}_k}{\partial P}$ becomes 
\begin{subequations}
\label{eq_derivative_Rk}
    \begin{align}
    \frac{\partial \tilde{R}_k}{\partial P}=&
    \frac{B}{\log(2) \left( 
    1+\frac{(M-K)\lambda\omega_k P \beta_k}{\sigma_k^2+\beta_k D}
    \right)
    }
    \frac{(M-K)\omega_k\beta_k}{\left(\sigma_k^2+\beta_k D\right)^2 }
    \label{eq_derivative_Rk_1}
    \\&
    \cdot \sqrt{\lambda}
    \left(1-e^{-\Psi}-\Psi e^{-\Psi} \right)
    \label{eq_derivative_Rk_2}
    \\&
    \left( \sigma_k^2-\frac{\sqrt{\pi}}{2}\beta_k \eta M P_{max} \frac{\textrm{erfc}(\sqrt{\Psi})}{\sqrt{\Psi}}  \right).
    \label{eq_derivative_Rk_3}
\end{align}
\end{subequations}
Observe that all variables $B,\sigma_k^2,\omega_k,P,\beta_k$ are non-negative. The same holds for $D$ as a property of power. Since the number of antennas must be greater than the number of UEs when employing ZF, $(M-K)$ is always positive. Finally, $\lambda$ is always non-negative by observation of squaring operation in (\ref{eq_lambda}). As such,  (\ref{eq_derivative_Rk_1}) creates a non-negative factor.  
Furthermore, \eqref{eq_derivative_Rk_2} is non-negative as well. Observe that $\sqrt{\lambda}$ is positive by the definition in \eqref{eq_lambda}. Next, the derivative of $1-e^{-\Psi}- \Psi e^{-\Psi}$ over $P$ equals $-\Psi^{2}e^{-\Psi}/P$, which is always a non-positive function of $P$. In other words, $1-e^{-\Psi}-\Psi e^{-\Psi}$ is monotonically decreasing function of $P$. The limit
    \begin{equation}
    \lim_{P \to \infty} 1-e^{-\Psi}-\Psi e^{-\Psi}=0,
    \end{equation}
implies that $1-e^{-\Psi}-\Psi e^{-\Psi}$ is non-negative as well. 
By taking only the (\ref{eq_derivative_Rk_3}) component and dividing it by a positive real number $\frac{\sqrt{\pi}}{2}\beta_k \eta M P_{max}$, a function $f_k(P)$ in (\ref{eq_f_k}) is defined. Therefore, the function $\frac{\partial \tilde{R}_k}{\partial P}$ has a root only when $f_k(P)=0$. 

The first derivative of $f_k(P)$ over $P$ can be calculated as 
\begin{equation}
    f'_{k}(P)=-\frac{\sqrt{\Psi}\textrm{erfc}(\sqrt{\Psi})}{2M P_{max}}-\frac{e^{-\Psi}}{\sqrt{\pi}P}.
    \label{eq_f_k_prim}
\end{equation}
Observe that $\textrm{erfc}(~)$ ranges from 1 to 0 for its positive arguments. Therefore, the function $f'_{k}(P)$, as a result of both its component being non-positive over the entire range of $P$ is non-positive itself. Since
\begin{equation}
    \lim_{P \to 0}\frac{\textrm{erfc}(\sqrt{\Psi})}{\sqrt{\Psi}}=0
\end{equation}
and
\begin{equation}
    \lim_{P \to \infty}\frac{\textrm{erfc}(\sqrt{\Psi})}{\sqrt{\Psi}}=\infty, 
    \label{eq_lim_f_k_prim}
\end{equation}
it means that the function $f_{k}(P)$ starts (for $P=0$) from value $\frac{2\sigma_k^2}{\sqrt{\pi}\beta_k \eta M P_{max}}$ and decreases till $-\infty$ for $P\to \infty$.
This means that there is only a single root and the function $f_k(P)$ is monotonically decreasing. 

\section{Proof of Lemma \ref{lemma_upper_lower_root}}
\label{sec:DAPA_proof_2}

It can be shown that the $\textrm{erfc}(x)$ function for non-negative arguments (as in our case in \eqref{eq_f_k})  has an upper bound $e^{-x^2}$ \cite{Chiani_erfc_2003} and lower bound $\sqrt{\frac{e}{2\pi}}e^{-2x^2}$\cite{Chang_erfc_2011}. Using the $\textrm{erfc}$ upper bound, the lower bound of function $f_k(P)$ can be defined as
\begin{equation}
f^{lower}_{k}(P)=\frac{2\sigma_k^2}{\sqrt{\pi}\beta_k \eta M P_{max}}-\frac{e^{-\Psi}}{\sqrt{\Psi}}.
\label{eq_f_k_lower}
\end{equation}
Since the function $f_k(P)$ is monotonically decreasing (see Lemma \ref{lem:DAPA_1}), it is expected that the root of its lower bound, i.e. the function $f^{lower}_{k}(P)$, will be smaller or equal to the root of $f_k(P)$, i.e., $\tilde{P}^{lower}_{k}\leq \tilde{P}_{k}: f^{lower}_{k}(\tilde{P}^{lower}_{k})=0 \land f_{k}(\tilde{P}_{k})=0$.
Fortunately, the root of the function $f^{lower}_{k}(P)$ can be found analytically, resulting in \eqref{eq_root_P_lower}.

A similar discussion can be carried out for the lower bound of \textrm{erfc}. It can be used to define an upper bound of $f_k(P)$ as

\begin{equation}
f^{upper}_{k}(P)=\frac{2\sigma_k^2}{\sqrt{\pi}\beta_k \eta M P_{max}}-\sqrt{\frac{e}{2\pi}}\frac{e^{-2\Psi}}{\sqrt{\Psi}}.
\label{eq_f_k_upper}
\end{equation}
Its analytical solution (\ref{eq_root_P_upper}) will be an upper bound for $\tilde{P}_{k}$.

\section{Proof of Lemma \ref{lem:DAPAD}}
\label{sec:DAPAD_proof}
To achieve the optimal total power distribution, KKT conditions, which provide the necessary conditions for the optimality of \eqref{eq:optprob5}, are applied. For a fixed $P$, the Lagrange dual function of \eqref{eq:optprob5} is defined as 
\begin{align}	L(\boldsymbol{\omega},\boldsymbol{\mu},\nu)=&-\sum_{k=1}^{K} B \log_2 \left(1 + \frac{(M-K)\lambda\omega_k P \beta_k}{\sigma^2_k + \beta_k D}\right) + \nonumber \\
&\sum_{k=1}^{K} \mu_k \left( - \omega_k \right) + 
\nu \left( 1 -\sum_{k=1}^{K} \omega_k  \right),
\end{align}
where $\boldsymbol{\mu}=[\mu_k], \forall k=1,\dots,K$ is the vector of Lagrange multipliers for the inequality constraints $\boldsymbol{\omega \geq 0}$, and $\nu$ is the Lagrange multiplier for the equality constraint in \eqref{eq:pc3}.

Let us assume that $\boldsymbol{\mu}^*=[\mu^*_k], \forall k=1,\dots,K$ and $\nu^*$ are the optimal solutions of the dual problem and $\boldsymbol{\omega}^*=[\omega^*_k], \forall k=1,\dots,K$ is the primal solution.
The KKT conditions can be listed as follows:
\begin{subequations}	
    \begin{align*}
    &\text{C1}:  \omega^*_k \geq 0, \forall k, \text{C2}: \sum_{k=1}^{K} \omega^*_k - 1 = 0, \\
    & \text{C3}: \mu^*_k \geq 0, \forall k,, \text{C4}: \mu^*_k (\omega^*_k) = 0, \forall k, \\
    &\text{C5}: \nabla_{\boldsymbol{\omega}^*} \left(-\sum_{k=1}^{K} B \log_2 \left(1 + \frac{(M-K) \lambda \omega_k^* P \beta_k}{\sigma^2_k + \beta_k D}\right)\right) \\
    &\quad - \sum_{k=1}^{K} \mu_k \nabla_{\boldsymbol{\omega}^*} (\omega^*_k) + \nu^* \nabla_{\boldsymbol{\omega}^*} \left(1- \sum_{k=1}^{K} \omega^*_k  \right) = 0.
    \end{align*}
\end{subequations}
Note that the condition C5, after computing the gradient,  can be reformulated as
$$\text{C5}:-\frac{B}{\frac{\sigma_k^2+ \beta_kD}{(M-K)\lambda P\beta_k} + \omega^*_k} - \mu^*_k + \nu^* = 0,~\forall k=1,\dots,K $$
and the set of KKT conditions C1-C5 can be solved directly to find $\boldsymbol{\omega}^*$, $\boldsymbol{\mu}^*$, and $\nu^*$. Observing C3, C4 and C5, the equivalent KKT conditions can be written as 
\begin{subequations}
    \begin{align*}
	&\text{E1}:\omega^*_k \geq 0,\forall k,~~~~~\text{E2}:\sum_{k=1}^{K} \omega^*_k = 1,\\
    &\text{E3}:\left(\nu^* -\frac{B}{\frac{\sigma_k^2+ \beta_kD}{(M-K)\lambda P\beta_k} + \omega_k^*} \right) \omega^*_k=0,\forall k,
	\\
	&\text{E4}:\nu^* \geq \frac{B}{\frac{\sigma_k^2+ \beta_kD}{(M-K)\lambda P\beta_k} + \omega_k^*} ,~\forall k=1,\dots,K,
	\end{align*}
 \end{subequations}
where E1 and E2 are equal to C1 and C2, respectively, and E3 is obtained by solving for $\mu_k^*$ in C5 and substitution in C4. 

Observe that the equivalent KKT conditions E1-E4 are independent of $\mu_k^*$ and therefore $\boldsymbol{\omega}^*$ can be computed as a function of $\nu^*$. Therefore, according to the first term 
in E3, we investigate when $\nu^* < \frac{B}{\frac{\sigma_k^2+ \beta_kD}{(M-K)\lambda P\beta_k}}$ and $\nu^* \geq \frac{B}{\frac{\sigma_k^2+ \beta_kD}{(M-K)\lambda P\beta_k}},~\forall k=1,\dots,K$ 
to find the optimal solution $\omega_k^*$ and conclude that the optimal solution $\omega_k^*$ is
\[   
\omega^*_k=
    \begin{cases}
        \frac{B}{\nu^*} - \frac{\sigma^2_k + \beta_kD}{(M-K)\lambda P\beta_k} &\quad \nu^* < \frac{B}{\frac{\sigma_k^2+ \beta_kD}{(M-K)\lambda P\beta_k}},\\
        0 &\quad \nu^* \geq \frac{B}{\frac{\sigma_k^2+ \beta_kD}{(M-K)\lambda P\beta_k}}, \\ 
    \end{cases}
\]
which can be reformulated as 
\begin{align}
\omega^*_k=\max \left\{ 0 , \frac{B}{\nu^*} - \frac{\sigma^2_k + \beta_kD}{(M-K)\lambda P\beta_k} \right\}.
\label{Optpower}
\end{align}

Note that by substituting the expression for $\omega^*_k$ into condition E2, we have $\sum_{k=1}^{K}  \max \left\{ 0 , \frac{B}{\nu^*} - \frac{\sigma^2_k + \beta_kD}{(M-K)\lambda P\beta_k} \right\} = 1$. The left-hand side is a piecewise linear increasing function of $\frac{B}{\nu^*}$, with breakpoints at $\frac{\sigma^2_k + \beta_kD}{(M-K)\lambda P\beta_k}$ for each $k$, so the equation has a unique solution which is readily determined. The optimal $\nu^*$ can be found using the bisection method.  However, optimized solutions also exist, for example in \cite{8995606}.

\ifCLASSOPTIONcaptionsoff
  \newpage
\fi

\bibliographystyle{IEEEtran}
\bibliography{IEEEabrv,bibtex/bib/IEEEexample}

\end{document}

%% file: system_model.tex
\pgfplotsset{compat=1.18}
\begin{tikzpicture}
\begin{scope}[yshift=0.3cm]
    \draw[->, thick] (2.48, 2.4) -- (4.0, 2.4); 
    \node at (3.2, 2.2) {\fontsize{8}{12}\selectfont frequency};
    \draw[->, thick] (2.48, 2.4) -- (2.48, 3.5); 
    \node[rotate=90] at (2.36, 2.8) {\fontsize{8}{12}\selectfont power};
    \node[draw, minimum width=1.4cm, minimum height=0.1cm, fill=red] at (3.18, 3.2){}; 
    \node at (4.13, 3.2) {\fontsize{8}{12}\selectfont $p_{1}$}; 
    \node[draw, minimum width=1.4cm, minimum height=0.1cm, fill=blue] at (3.18, 3.0){};
    \node at (4.13, 3.0) {\fontsize{8}{12}\selectfont $p_2$}; 
    \node at (3.1, 2.8) {$\vdots$};
    \node[draw, minimum width=1.4cm, minimum height=0.1cm, fill=green] at (3.18, 2.52){}; 
    \node at (4.13, 2.52){\fontsize{8}{12}\selectfont \text{$p_K$}}; 
\end{scope}
    
    \node[rectangle, draw, minimum width=1cm, minimum height=5cm, align=center] (box) at (0,0) {};
    
    \node[rotate=270] at (0,0) {DSP + DAC + UP Conversion};

    \draw[->] (0.5, 1.5) -- ++(1, 0); 
    \node[rectangle, draw, minimum width=0.8cm, minimum height=0.4cm] at (2.02, 1.5) {
    \begin{tikzpicture}
        \draw[domain=2.7:3, samples=50, color=black] plot (\x, \x); 
        \draw[black] (3, 3) -- (3.5,3);
    \end{tikzpicture}}; 
    \draw[-] (2.55, 1.5) -- ++(0.5,0);
    \draw[-] (3.05, 1.5) -- (3.05, 2.0);
    \draw[-] (3.05, 1.7) -- (3.3, 2.0);
    \draw[-] (3.05, 1.7) -- (2.8, 2.0);

    \draw[->] (0.5, 0) -- ++(1, 0); 
    \node[rectangle, draw, minimum width=0.8cm, minimum height=0.4cm] at (2.02, 0) {
    \begin{tikzpicture}
        \draw[domain=2.7:3, samples=50, color=black] plot (\x, \x); 
        \draw[black] (3, 3) -- (3.5,3);
    \end{tikzpicture}}; 
    \draw[-] (2.55, 0) -- ++(0.5,0);
    \draw[-] (3.05, 0) -- (3.05, 0.5);
    \draw[-] (3.05, 0.2) -- (3.3, 0.5);
    \draw[-] (3.05, 0.2) -- (2.8, 0.5);

    \draw[->] (0.5, -1.5) -- ++(1, 0); 
    \node[rectangle, draw, minimum width=0.8cm, minimum height=0.4cm] at (2.02, -1.5) {
    \begin{tikzpicture}
        \draw[domain=2.7:3, samples=50, color=black] plot (\x, \x); 
        \draw[black] (3, 3) -- (3.5,3);
    \end{tikzpicture}}; 
    \draw[-] (2.55, -1.5) -- ++(0.5,0);
    \draw[-] (3.05, -1.5) -- (3.05, -1.0);
    \draw[-] (3.05, -1.3) -- (3.3, -1.0);
    \draw[-] (3.05, -1.3) -- (2.8, -1.0);

\node at (3.0, 1.3) {\fontsize{8}{12}\selectfont $1$};
\node at (2.1, 0.9) {\fontsize{8}{12}\selectfont $PA_1$};

\node at (3.0, -0.2) {\fontsize{8}{12}\selectfont $2$};
\node at (2.1, -0.6) {\fontsize{8}{12}\selectfont $PA_2$};

\node at (3.0, -0.5) {$\vdots$};

\node at (3.0, -1.7) {\fontsize{8}{12}\selectfont $M$};
\node at (2.1, -2.1) {\fontsize{8}{12}\selectfont $PA_M$};

\draw[->, red](3.3, 2.1) -- (4.5, 2.5);
\node at (4.4, 2.1) {\color{red}{$\beta_1$}};
\draw[->, blue](3.3, 1.9) -- (4.5, 1.1);
\node at (4.4, 0.7) {\color{blue}{$\beta_2$}};
\node at (4.4, 0.3) {$\vdots$};
\draw[->, green](3.2, 1.7) -- (4.6, -0.6);
\node at (4.4, -0.8) {\color{green}{$\beta_K$}};

    \begin{scope}[xshift=4cm] 
        \node[rectangle, draw, minimum width=0.4cm, minimum height=0.4cm] at (2.52, 2.19) {
            \begin{tikzpicture}
                 \node at (2.52, 1.5) {+};
                 \draw (2.52, 1.5) circle (0.15);
                 \draw[-] (2.66, 1.5) -- (3.2, 1.5);
                 \node at (3.5, 1.5) {\fontsize{8}{12}\selectfont RX};                 
            \end{tikzpicture}
        }; 
        \draw[<-] (2.0, 2.0) -- (2.0, 1.6);
        \node at (2.0, 1.45) {\fontsize{8}{12}\selectfont $\sigma_1^2$};
        \draw[-] (1.0, 2.2) -- (1.85, 2.2);
        \draw[-] (1.0, 2.2) -- (1.0, 2.7);
        \draw[-] (1.0, 2.4) -- (1.3, 2.7);
        \draw[-] (1.0, 2.4) -- (0.7, 2.7);

        \node[rectangle, draw, minimum width=0.4cm, minimum height=0.2cm] at (2.52, 0.60) {
            \begin{tikzpicture}
                \node at (2.52, 0) {+};
                \draw (2.52, 0) circle (0.15);
                \draw[-] (2.66, 0) -- (3.2, 0);
                \node at (3.5, 0) {\fontsize{8}{12}\selectfont RX};                              
            \end{tikzpicture}
        }; 
        \draw[<-] (2.0, 0.4) -- (2.0, 0.0);  
        \node at (2.0, -0.2) {\fontsize{8}{12}\selectfont $\sigma_2^2$};
        \draw[-] (1.0, 0.6) -- (1.85, 0.6);
        \draw[-] (1.0, 0.6) -- (1.0, 1.1);
        \draw[-] (1.0, 0.8) -- (1.3, 1.1);
        \draw[-] (1.0, 0.8) -- (0.7, 1.1);
        
\node at (2.5, -0.3) {$\vdots$};

        \node[rectangle, draw, minimum width=0.4cm, minimum height=0.4cm] at (2.52, -1.3) {
            \begin{tikzpicture}
               \node at (2.52, -1.5) {+};
                \draw (2.52, -1.5) circle (0.15);
                 \draw[-] (2.66, -1.5) -- (3.2, -1.5);
                 \node at (3.5, -1.5) {\fontsize{8}{12}\selectfont RX};                
            \end{tikzpicture}
        }; 
        \draw[<-] (2.0, -1.5) -- (2.0, -1.9);
         \node at (2.0, -2.1) {\fontsize{8}{12}\selectfont $\sigma_K^2$};
        \draw[-] (1.0, -1.33) -- (1.87, -1.33);
        \draw[-] (1.0, -1.33) -- (1.0, -0.8);
        \draw[-] (1.0, -1.11) -- (1.3, -0.8);
        \draw[-] (1.0, -1.11) -- (0.7, -0.8);
    \end{scope}

\end{tikzpicture}


%% file: bare_jrnl_RR1.bbl
\begin{thebibliography}{10}
\providecommand{\url}[1]{#1}
\csname url@samestyle\endcsname
\providecommand{\newblock}{\relax}
\providecommand{\bibinfo}[2]{#2}
\providecommand{\BIBentrySTDinterwordspacing}{\spaceskip=0pt\relax}
\providecommand{\BIBentryALTinterwordstretchfactor}{4}
\providecommand{\BIBentryALTinterwordspacing}{\spaceskip=\fontdimen2\font plus
\BIBentryALTinterwordstretchfactor\fontdimen3\font minus \fontdimen4\font\relax}
\providecommand{\BIBforeignlanguage}[2]{{%
\expandafter\ifx\csname l@#1\endcsname\relax
\typeout{** WARNING: IEEEtran.bst: No hyphenation pattern has been}%
\typeout{** loaded for the language `#1'. Using the pattern for}%
\typeout{** the default language instead.}%
\else
\language=\csname l@#1\endcsname
\fi
#2}}
\providecommand{\BIBdecl}{\relax}
\BIBdecl

\bibitem{Bjornson_2018_MIMO_unlimited}
E.~Björnson, J.~Hoydis, and L.~Sanguinetti, ``{Massive {MIMO} Has Unlimited Capacity},'' \emph{IEEE Transactions on Wireless Communications}, vol.~17, no.~1, pp. 574--590, 2018.

\bibitem{ochiai2013analysis}
H.~Ochiai, ``{An Analysis of Band-Limited Communication Systems from Amplifier Efficiency and Distortion Perspective},'' \emph{IEEE Transactions on Communications}, vol.~61, no.~4, pp. 1460--1472, 2013.

\bibitem{Kryszkiewicz_TR_2018}
P.~Kryszkiewicz, ``{Amplifier-Coupled Tone Reservation for Minimization of {OFDM} Nonlinear Distortion},'' \emph{IEEE Transactions on Vehicular Technology}, vol.~67, no.~5, pp. 4316--4324, 2018.

\bibitem{Dinis_Qoptimal_RX_2024}
D.~Dinis, D.~Costa, J.~Guerreiro, J.~Madeira, and R.~Dinis, ``{Quasi-Optimum Detection of Nonlinear {OFDM}: Performance Bounds and Receiver Design},'' \emph{IEEE Transactions on Communications}, pp. 1--1, 2024.

\bibitem{Baghani_IMD3_power_allocation_wideband_2014}
M.~Baghani, A.~Mohammadi, M.~Majidi, and M.~Valkama, ``{Analysis and Rate Optimization of {OFDM}-Based Cognitive Radio Networks Under Power Amplifier Nonlinearity},'' \emph{IEEE Transactions on Communications}, vol.~62, no.~10, pp. 3410--3419, 2014.

\bibitem{lee2014characterization}
T.~Lee and H.~Ochiai, ``{Characterization of Power Spectral Density for Nonlinearly Amplified {OFDM} Signals based on Cross-correlation Coefficient},'' \emph{EURASIP Journal on Wireless Communications and Networking}, vol. 2014, pp. 1--15, 2014.

\bibitem{Amirhossein_2019_power_allocation_GFDM}
A.~Mohammadian, M.~Baghani, and C.~Tellambura, ``{Optimal Power Allocation of {GFDM} Secondary Links With Power Amplifier Nonlinearity and {ACI}},'' \emph{IEEE Wireless Communications Letters}, vol.~8, no.~1, pp. 93--96, 2019.

\bibitem{10697405}
X.~Wang, S.~Bi, X.~Li, X.~Lin, Z.~Quan, and Y.-J.~A. Zhang, ``Capacity analysis and throughput maximization of noma with non-linear power amplifier distortion,'' \emph{IEEE Transactions on Wireless Communications}, vol.~23, no.~12, pp. 18\,331--18\,345, 2024.

\bibitem{bjornson2014massive}
E.~Bj{\"o}rnson, J.~Hoydis, M.~Kountouris, and M.~Debbah, ``{Massive {MIMO} Systems with Non-Ideal Hardware: Energy Efficiency, Estimation, and Capacity Limits},'' \emph{IEEE Transactions on information theory}, vol.~60, no.~11, pp. 7112--7139, 2014.

\bibitem{mollen_spatial_char}
C.~Mollén, U.~Gustavsson, T.~Eriksson, and E.~G. Larsson, ``{Spatial Characteristics of Distortion Radiated From Antenna Arrays With Transceiver Nonlinearities},'' \emph{IEEE Transactions on Wireless Communications}, vol.~17, no.~10, pp. 6663--6679, 2018.

\bibitem{Salman_OFDM_MIMO_PA_modelin_2023}
M.~B. Salman, E.~Björnson, G.~M. Güvensen, and T.~Çiloğlu, ``{Analytical Nonlinear Distortion Characterization for Frequency-Selective Massive {MIMO} Channels},'' in \emph{ICC 2023 - IEEE International Conference on Communications}, 2023, pp. 6535--6540.

\bibitem{Liu_2021_MMIMO_impairments}
Z.~Liu, C.-H. Lee, W.~Xu, and S.~Li, ``{Energy-Efficient Design for Massive {MIMO} With Hardware Impairments},'' \emph{IEEE Transactions on Wireless Communications}, vol.~20, no.~2, pp. 843--857, 2021.

\bibitem{persson2013amplifier}
D.~Persson, T.~Eriksson, and E.~G. Larsson, ``{Amplifier-Aware Multiple-Input Multiple-Output Power Allocation},'' \emph{IEEE Communications Letters}, vol.~17, no.~6, pp. 1112--1115, 2013.

\bibitem{Liu_2024_allocation_DMIMO}
B.~Liu and S.~Pollin, ``{Distortion-Aware Power Allocation for Multi-Stream Distributed Massive {MIMO} System With Nonlinear Power Amplifier},'' \emph{IEEE Open Journal of the Communications Society}, vol.~5, pp. 1566--1578, 2024.

\bibitem{hoffmann2023contextual}
M.~Hoffmann and P.~Kryszkiewicz, ``Contextual bandit-based amplifier ibo optimization in massive mimo network,'' \emph{IEEE Access}, vol.~11, pp. 127\,035--127\,042, 2023.

\bibitem{Jee_2021_MISO_PA_power_allocation}
J.~Jee, G.~Kwon, and H.~Park, ``{Joint Precoding and Power Allocation for Multiuser {MIMO} System With Nonlinear Power Amplifiers},'' \emph{IEEE Transactions on Vehicular Technology}, vol.~70, no.~9, pp. 8883--8897, 2021.

\bibitem{Wachowiak2023}
M.~Wachowiak and P.~Kryszkiewicz, ``{Clipping Noise Cancellation Receiver for the Downlink of Massive {MIMO OFDM} System},'' \emph{IEEE Transactions on Communications}, vol.~71, no.~10, pp. 6061--6073, 2023.

\bibitem{Nokia_3gpp_Rapp}
Nokia, ``{Realistic power amplifier model for the New Radio evaluation},'' 3GPP doc. {R4-163314}, May 2016.

\bibitem{raich2005_optimal_nonlinearity}
R.~Raich, H.~Qian, and G.~T. Zhou, ``{Optimization of {SNDR} for Amplitude-Limited Nonlinearities},'' \emph{IEEE Transactions on Communications}, vol.~53, no.~11, pp. 1964--1972, 2005.

\bibitem{Joung_2015_survey_EE_PA}
J.~Joung, C.~K. Ho, K.~Adachi, and S.~Sun, ``{A Survey on Power-Amplifier-Centric Techniques for Spectrum- and Energy-Efficient Wireless Communications},'' \emph{IEEE Communications Surveys \& Tutorials}, vol.~17, no.~1, pp. 315--333, 2015.

\bibitem{Wei_2010_dist_OFDM}
S.~Wei, D.~L. Goeckel, and P.~A. Kelly, ``{Convergence of the Complex Envelope of Bandlimited {OFDM} Signals},'' \emph{IEEE Transactions on Information Theory}, vol.~56, no.~10, pp. 4893--4904, 2010.

\bibitem{Nossek2011power}
A.~Mezghani and J.~A. Nossek, ``{Power Efficiency in Communication Systems from a circuit Perspective},'' in \emph{2011 IEEE International Symposium of Circuits and Systems (ISCAS)}.\hskip 1em plus 0.5em minus 0.4em\relax IEEE, 2011, pp. 1896--1899.

\bibitem{kryszkiewicz2023efficiency}
P.~Kryszkiewicz, ``{Efficiency Maximization for Battery-Powered {OFDM} Transmitter via Amplifier Operating Point Adjustment},'' \emph{Sensors}, vol.~23, no.~1, p. 474, 2023.

\bibitem{massivemimobook}
\BIBentryALTinterwordspacing
E.~Bj\"{o}rnson, J.~Hoydis, and L.~Sanguinetti, ``{Massive {MIMO} Networks: {Spectral}, Energy, and Hardware Efficiency},'' \emph{Foundations and Trends{\textregistered} in Signal Processing}, vol.~11, no. 3-4, pp. 154--655, 2017. [Online]. Available: \url{http://dx.doi.org/10.1561/2000000093}
\BIBentrySTDinterwordspacing

\bibitem{itu-m2135-1-2009}
\BIBentryALTinterwordspacing
{International Telecommunication Union}, ``{Guidelines for Evaluation of Radio Interface Technologies for IMT-Advanced},'' International Telecommunication Union, Tech. Rep. M.2135-1, 2009. [Online]. Available: \url{https://www.itu.int/dms_pub/itu-r/opb/rep/r-rep-m.2135-1-2009-pdf-e.pdf}
\BIBentrySTDinterwordspacing

\bibitem{3gpp_38141}
3GPP, ``"3rd generation partnership project;technical specification group radio access network; nr; base station (bs) conformance testing part 1: Conducted conformance testing (release 18)",'' 3GPP, Tech. Rep. TS 38.141-1 V18.7.0, 2024.

\bibitem{8641436}
K.~Senel, E.~Björnson, and E.~G. Larsson, ``Joint transmit and circuit power minimization in massive mimo with downlink sinr constraints: When to turn on massive mimo?'' \emph{IEEE Transactions on Wireless Communications}, vol.~18, no.~3, pp. 1834--1846, 2019.

\bibitem{boyd2004convex}
S.~Boyd and L.~Vandenberghe, \emph{{Convex Optimization}}.\hskip 1em plus 0.5em minus 0.4em\relax Cambridge university press, 2004.

\bibitem{Chiani_erfc_2003}
M.~Chiani, D.~Dardari, and M.~Simon, ``{New Exponential Bounds and Approximations for the Computation of Error Probability in Fading Channels},'' \emph{IEEE Transactions on Wireless Communications}, vol.~2, no.~4, pp. 840--845, 2003.

\bibitem{Chang_erfc_2011}
S.-H. Chang, P.~C. Cosman, and L.~B. Milstein, ``{Chernoff-Type Bounds for the {Gaussian} Error Function},'' \emph{IEEE Transactions on Communications}, vol.~59, no.~11, pp. 2939--2944, 2011.

\bibitem{8995606}
C.~Xing, Y.~Jing, S.~Wang, S.~Ma, and H.~V. Poor, ``{New Viewpoint and Algorithms for Water-Filling Solutions in Wireless Communications},'' \emph{IEEE Transactions on Signal Processing}, vol.~68, pp. 1618--1634, 2020.

\end{thebibliography}
